\documentclass[11pt]{article}
\usepackage{authblk}
\title{Frugal, Flexible,  Faithful: Causal Data Simulation via Frengression}
\author[1]{Linying Yang\thanks{E-mails:
    \texttt{linying.yang@stats.ox.ac.uk},
    \texttt{evans@stats.ox.ac.uk},
    \texttt{xwshen@uw.edu}}}
\author[2]{Jens Magelund Tarp}
\author[1,3]{Robin J.~Evans\protect\footnotemark[1]}
\author[4]{Xinwei Shen\protect\footnotemark[1]}
\affil[1]{Department of Statistics, University of Oxford}
\affil[2]{Novo Nordisk}
\affil[3]{Pioneer Centre for SMARTbiomed, University of Oxford}
\affil[4]{Department of Statistics, University of Washington}
\date{\today}
\usepackage{amsmath, amsfonts, amsthm, amssymb, bbm, bm}
\usepackage{mathrsfs}
\usepackage{centernot}
\usepackage{enumerate}
\usepackage{tabularx}
\usepackage{multicol}
\usepackage{multirow}

\renewcommand{\Pr}{\mathbb{P}}
\newcommand{\E}{\mathbb{E}}
\DeclareMathOperator{\Var}{Var}

\newcommand{\cmid}{\,|\,}

\newcommand\indep{\protect\mathpalette{\protect\independenT}{\perp}}
\def\independenT#1#2{\mathrel{\rlap{$#1#2$}\mkern2mu{#1#2}}}

\theoremstyle{plain}

\theoremstyle{definition}

\newcommand{\benum}{\begin{enumerate}}
\newcommand{\eenum}{\end{enumerate}}

\newcommand{\bitem}{\begin{itemize}}
\newcommand{\eitem}{\end{itemize}}

\newcommand{\barr}{\begin{array}}
\newcommand{\earr}{\end{array}}

\newcommand{\bmat}{\begin{pmatrix}}
\newcommand{\emat}{\end{pmatrix}}

\newcommand{\blist}{\renewcommand{\labelenumi}{\textbf{\arabic{enumi}}.} \begin{enumerate}}
\newcommand{\elist}{\end{enumerate} \renewcommand{\labelenumi}{\arabic{enumi}.}}

\def\bal#1\eal{\begin{align*}#1\end{align*}}

\usepackage{tikz}
\usetikzlibrary{shapes, arrows, arrows.meta, calc, positioning}

\tikzset{nv/.style={circle, color=red, fill=red, inner sep=0.5mm}}
\tikzset{rv/.style={circle, minimum size=6.5mm, inner sep=0.5mm}}
\tikzset{fv/.style={rectangle, draw, thick, minimum size=6mm, inner sep=0.5mm}}
\tikzset{lv/.style={circle, color=red, fill=gray!30, draw, thick, minimum size=6.5mm, inner sep=0.5mm}}
\tikzset{sv/.style={circle, double, draw, thick, minimum size=6.5mm, inner sep=0.5mm}}
\tikzset{rve/.style={ellipse, minimum size=6.5mm, inner sep=0.5mm}}

\tikzset{rvs/.style={circle, minimum size=5.5mm, inner sep=0.5mm}}
\tikzset{fvs/.style={rectangle, draw, thick, minimum size=5mm, inner sep=0.5mm}}
\tikzset{lvs/.style={circle, color=red, fill=gray!30, draw, thick, minimum size=5.5mm, inner sep=0.5mm}}
\tikzset{svs/.style={circle, double, draw, thick, minimum size=5.5mm, inner sep=0.5mm}}
\tikzset{rves/.style={ellipse, minimum size=5.5mm, inner sep=0.5mm}}

\tikzset{deg/.style={->, very thick, color=blue}}
\tikzset{sdeg/.style={*->, very thick, color=blue}}
\tikzset{degl/.style={->, very thick, color=red}}
\tikzset{beg/.style={<->, very thick, color=red}}
\tikzset{cdeg/.style={{Circle[length=+2pt 2.5,width=+2pt 2.5, fill=none]}->, very thick, color=blue}}
\tikzset{cceg/.style={{Circle[length=+2pt 2.5,width=+2pt 2.5, fill=none]}-{Circle[length=+2pt 2.5,width=+2pt 2.5, fill=none]}, very thick}}
\tikzset{uceg/.style={{Circle[length=+2pt 2.5,width=+2pt 2.5, fill=none]}-, very thick}}
\tikzset{ueg/.style={very thick}}

\usepackage{amsmath, amssymb, amsthm, mathrsfs, bm}
\usepackage{geometry, extarrows, graphicx, enumitem, booktabs, subfigure}
\usepackage{natbib}
\usepackage{titlesec}
\usepackage{hyperref}
\usepackage{cleveref}
\usepackage{array}
\usepackage{booktabs,makecell}

\makeatletter
\let\oldappendix\appendix
\renewcommand{\appendix}{%
  \oldappendix
  \crefalias{section}{appendix}%
  \crefalias{subsection}{appendix}%
  \crefalias{subsubsection}{appendix}%
}
\makeatother

\crefname{appendix}{Appendix}{Appendices}
\Crefname{appendix}{Appendix}{Appendices}

\usepackage{fullpage}

\DeclareMathOperator*{\argmin}{argmin}

\def\E{\mathbb{E}}
\def\R{\mathbb{R}}
\def\P{\mathbb{P}}
\def\Q{\mathbb{Q}}

\def\cL{\mathcal{L}}
\def\cN{\mathcal{N}}
\def\cX{\mathcal{X}}

\DeclareMathOperator{\ed}{ED}
\DeclareMathOperator{\es}{ES}
\DeclareMathOperator{\mmd}{MMD}

\newcommand{\YIZX}{{Y\hspace{-.5pt}|\hspace{-.75pt}Z \hspace{-.75pt} X}}
\newcommand{\ZX}{{\hspace{-.5pt} Z \hspace{-1pt} X}}
\newcommand{\Xpaz}{X_0}
\newcommand{\Xchz}{X_1}

\newcommand{\ZIXp}{Z\hspace{-.5pt}|\hspace{-.5pt}\Xpaz\hspace{-1pt}}

\newcommand{\XpZXc}{\Xpaz\hspace{-.5pt}Z\hspace{-1pt}\Xchz\hspace{-1.5pt}}

\newcommand{\XpZ}{\Xpaz\hspace{-.5pt}Z}

\newcommand{\Yx}{Y\hspace{-.75pt}(x)}

\theoremstyle{plain}
\newtheorem{theorem}{Theorem}[section]
\newtheorem{proposition}{Proposition}[section]
\newtheorem{lemma}{Lemma}[section]
\newtheorem{corollary}{Corollary}[section]
\newtheorem{assumption}{Assumption}[section]

\theoremstyle{remark}

\newtheorem{remark}{Remark}[section]

\usepackage{colortbl}
\definecolor{peach}{RGB}{255,235,225}
\definecolor{lightblue}{RGB}{224,243,249}
\definecolor{lightgreen}{RGB}{230,255,230}

\begin{document}

\maketitle

\begin{abstract}
Machine learning has expanded the scope of causal inference through flexible function classes, but evaluating such methods remains difficult because the true causal data-generating mechanism and interventional structure are typically unknown.
We introduce frengression, a generative framework that directly models marginal interventional distributions. The key idea is to represent these distributions via a reparameterised noise variable, yielding a generative model that supports sampling under user-specified interventions. In contrast to approaches that infer causal quantities from observational conditionals, frengression targets the interventional estimand itself and offers a principled route to generative modelling of interventional distributions. Frengression provides a coherent procedure for both estimation and simulation in multivariate and time-varying settings. We establish consistency and characterise extrapolation behaviour of the proposed estimator. Empirical results on clinical trial data illustrate its practical performance.
\end{abstract}

\section{Introduction} \label{sec:introduction}

The use of machine learning tools has given causal inference a new lease of life, enabling complex models to be used with principled causal estimators and guarantees about statistically important quantities \citep{wager2018estimation,chernozhukov2018double, hahn2020bayesian}.  To build trustworthy causal models, however, we also need to understand when these methods may be more or less reliable, or perhaps fail completely. This implies that causal inference needs a set of good benchmarking tools. Unfortunately, real‑world datasets are not ideal for this task, because they cannot give us access to the ground truth other than in a few very special circumstances.  In particular, they rarely provide the counterfactual outcomes we care about, and the distribution we want to evaluate often differs from the one that produced the observations. Well‑designed simulations can address this discrepancy \citep{neal2020realcause,parikh2022validating}; they allow us to choose a ground truth, stress‑test new methods, compare their generalisability and stability, and expose failure modes before deployment.  In trial design, simulations also guide sample size calculations, making patient recruitment more efficient. Valid benchmarks have therefore become central to method development in the AI era. Faithful simulation is no longer optional; it is a core step in building trustworthy causal models \citep{friedrich2024role, pezoulas2024synthetic, gamella2025causal}.

Generating data from a marginal interventional distribution, rather than merely reproducing observational conditional relationships, poses a distinct modelling problem \citep{havercroft2012simulating}. Let $Z\in\R^{d_z}$ denote covariates, $X\in\R^{d_x}$ treatment, and $Y\in\R^{d_y}$ outcome. Inspired by \citet{evans2024parameterizing}, the joint distribution $\P$ of $(Z,X,Y)$ can be decomposed into the distribution of the past $\P_{\ZX}$, the interventional distribution $\P_{\Yx}$, and a remaining component that completes the joint distribution. This decomposition isolates the causal margin $\P_{\Yx}$ as a primary object.

The difficulty is that $\P_{\Yx}$ cannot, in general, be obtained from the structural equation of the outcome alone: while $Y\mid X,Z$ is specified by a structural model, the interventional margin additionally depends on the distribution of $Z$ under intervention. Thus, modelling $\P_{\Yx}$ requires capturing observational structure while enabling controlled generation under interventions.

We introduce \emph{frengression}, a generative approach that models the marginal interventional distribution directly. The central idea is to represent $\P_{\Yx}$ via a reparameterised independent noise variable, yielding a generative model that supports sampling under user-specified interventions.

\subsection{Contribution}
\label{sec:contribution}
Frengression directly targets marginal causal quantities while also generating complex benchmark datasets with user-specified marginal causal distributions. This unified framework guarantees a precisely defined marginal effect while closely reflecting real-world conditions. Consistent with the frugal parameterisation, frengression consists of variation-independent components to isolate and control specific causal properties, such as the strength of confounding.

We highlight some key properties of frengression:
\begin{itemize}
    \item \textbf{Generative model for causal margin.} Frengression is a flexible deep generative model for distributional regression that represents the entire outcome distribution while explicitly specifying the marginal interventional distribution of interest. Directly targeting the marginal quantity of interest allows for better estimation and more precise control. 
    \item \textbf{Broad applicability.} Our framework supports both estimation and simulation for a variety of (complex) settings, including distributional treatment effects, continuous treatments, longitudinal settings, and survival data.
    \item \textbf{Flexibility and robustness.} Frengression factorises the joint distribution according to a frugal, variation-independent parameterisation, enabling each component to be modelled with any suitable method (not limited to neural networks). The model fitting does not require extensive hyperparameter tuning, enhancing overall stability. 
    \item \textbf{Theoretical guarantees.}  We provide a consistency guarantee for models trained with the finite-sample energy score, along with extrapolation guarantees under continuous treatments.
\end{itemize}
\subsection{Related Work}
\label{sec:related_work}
We review existing approaches to generative modelling for causal estimation and simulation and how they motivate our frengression framework.  We summarise this comparison between frengression and other methods in \Cref{tab:methods}.

\begin{table}[ht]
  \centering
  \scriptsize 
  \resizebox{\textwidth}{!}{%
    \begin{tabular}{c  c   c  c c  c c}
      \toprule
      \rowcolor{white}
      \textbf{Method} 
        & \makecell{\textbf{Binary}\\\textbf{Treatment}}
        & \makecell{\textbf{Continuous}\\\textbf{Treatment}}
        & \makecell{\textbf{Longitudinal}\\\textbf{Data}}
        & \makecell{\textbf{Distributional}\\\textbf{Treatment Effects}} &
        \textbf{Extrapolation}
        & \makecell{\textbf{Modelling}\\\textbf{Causal Margin}}
        \\
        
      \midrule

      \makecell{TARNet \\ CFRNet \\ Dragonnet\\ CEVAE}
       
        & \checkmark 
        & \textcolor{red}{$\times$} 
        & \textcolor{red}{$\times$} 
        & \textcolor{red}{$\times$} 
        & \textcolor{red}{$\times$}
        & \textcolor{red}{$\times$}\\
        \hline
      \makecell{CausalEGM \\ CausalBGM
       }
        & \checkmark 
        & \checkmark 
        & \textcolor{red}{$\times$} 
        & \textcolor{red}{$\times$}
        & \textcolor{red}{$\times$}
        & \textcolor{red}{$\times$} \\
       \hline
       \makecell{RealCause \\ Credence}
       
        & \checkmark 
        & \textcolor{red}{$\times$} 
        & \textcolor{red}{$\times$} 
        & \textcolor{red}{$\times$} 
        & \textcolor{red}{$\times$}
        & \textcolor{red}{$\times$}\\
        \hline
        DeepACE
        & \checkmark 
        &  \textcolor{red}{$\times$}
         & only static treatment regime
        & \textcolor{red}{$\times$} 
         & \textcolor{red}{$\times$}
         & \textcolor{red}{$\times$}\\
        \hline
        Deep LTMLE
        & \checkmark 
        & \textcolor{red}{$\times$} & \checkmark
        & \textcolor{red}{$\times$} 
         & \textcolor{red}{$\times$}
         & \textcolor{red}{$\times$}\\
        \hline
      Frugal Flow
        & \checkmark
        & \checkmark & \textcolor{red}{$\times$} 
        & \textcolor{red}{$\times$} 
         & \textcolor{red}{$\times$}
         & \checkmark\\
    \hline
      \textbf{Frengression}
        & \textcolor{teal}{\checkmark}  & \textcolor{teal}{\checkmark} 
        & \textcolor{teal}{\checkmark} & \textcolor{teal}{\checkmark} & \textcolor{teal}{\checkmark} & \textcolor{teal}{\checkmark}\\
      \bottomrule
    \end{tabular}%
  }
  \caption{Comparison of methods.}
  \label{tab:methods}
\end{table}

\subsubsection{(Deep) generative models for the causal margin}

\label{sec:generative_causal_margin}
Deep learning models have shown remarkable flexibility for causal estimation in high-dimensional,
complex settings. TARNet and CFRNet learn shared representations balancing treated and control groups \citep{shalit2017estimating}; Dragonnet incorporates propensity score estimation into the outcome network \citep{shi2019adapting}. However, since they enforce covariate balance across treatment arms, these architectures are restricted to binary treatments and static covariates, and do not extend to longitudinal data.

Structural equation models \citep{pearl2009causality} provide a natural generative framework: one builds a generative model directly for the structural equations of each variable. \citet{holovchak2025distributional} applied this idea in the nonparametric instrumental variable setting, estimating the full (conditional) interventional distribution under latent confounding. Other generative approaches to causal effect estimation include CEVAE \citep{louizos2017causal}, which uses variational autoencoders to recover the joint distribution of observed and latent confounders from purely observational data; CausalEGM \citep{liu2024encoding}, which encodes confounders into a low-dimensional latent space via an encoder–generator architecture, accommodating continuous treatments (e.g.~dosage levels); and its extension CausalBGM \citep{liu2025ai}, which adds a Bayesian neural network returning both the mean and variance of individual treatment-effect posteriors under a Gaussian assumption.

All these methods model the \emph{conditional} distribution. Modelling the \emph{marginal} interventional distribution generatively is less direct, since the causal margin is not in general induced by the structural equation of the outcome alone. Concretely, if $Y=f(X,Z,\varepsilon_Y)$ with $\varepsilon_Y$ exogenous, sampling from $\P_{\Yx|Z=z}$ requires only drawing $\varepsilon_Y$ and computing $f(x,z,\varepsilon_Y)$; sampling from the marginal $\P_{\Yx}$ additionally requires simulation from the interventional distribution of $Z(x)$. The generative mechanism inducing the causal margin therefore does not coincide with the original structural form. Frengression instead constructs a generative model for the causal margin directly, via a reparameterised independent noise variable. The variation independence of its components permits flexible specification of the margin while preserving causal validity.

\subsubsection{Modelling longitudinal data}

Causal estimation in longitudinal settings remains comparatively underdeveloped. Building on LTMLE (Longitudinal Targeted Maximum Likelihood Estimation, \citealp{van2012targeted,lendle2017ltmle}) and iterative g-computation, \citet{frauen2023estimating} proposed DeepACE, which uses recurrent neural networks to model patient trajectories and estimate average risk ratios; it handles only static treatment regimes. \citet{shirakawa2024longitudinal} introduced Deep LTMLE, extending LTMLE and DeepACE through transformer architectures to estimate counterfactual mean outcomes under dynamic treatment policies. Both methods estimate marginal causal quantities by empirical averaging of conditional potential outcomes, require careful hyperparameter tuning, and do not support data simulation.

By contrast, frengression extends directly to longitudinal data, accommodating marginal and conditional estimands under both static and dynamic regimes, and supporting data simulation with little tuning.

\subsubsection{Modelling distributional treatment effects}
Most causal estimation methods target the average treatment effect, but distributional treatment effects (DTE)  have become increasingly important for richer insights beyond the mean \citep{abadie2002bootstrap,bitler2017can}. Existing approaches include: semiparametric models \citep{diaz2017efficient, kennedy2023semiparametric}, which require correct model specification; kernel-based methods, which avoid parametric assumptions but are difficult to interpret and scale poorly in high dimensions \citep{muandet2017kernel,park2021conditional}; neural network methods \citep{holovchak2025distributional}, which are flexible but do not admit user-specified causal quantities; and Bayesian methods \citep{Venturini2015quantile,xu2018bayesian}, which are computationally expensive. Each trades off some combination of flexibility, interpretability, and efficiency.

Frengression bridges this gap through its modular, variation-independent parameterisation and distributional regression models, which directly target specific distributional quantities while preserving both transparency and efficiency in estimation and data simulation.

\subsubsection{Generative models for causal data simulation}
Several generative models simulate data for benchmarking causal inference methods. \citet{athey2024using} use Wasserstein generative adversarial networks (GANs) in a Monte Carlo simulation framework. \citet{neal2020realcause} developed RealCause, based on neural autoregressive flows \citep{huang2018neural}. \citet{parikh2022validating} proposed Credence, a variational autoencoder factorising the joint distribution into the covariate marginal, the treatment given covariates, and the outcome given covariates and treatment. The conditional and marginal outcome distributions are specified only indirectly through a post-hoc penalty in the training loss, and the model is restricted to binary treatments.

None of these methods allows direct control of the marginal causal effect. Frugal flows \citep{de2024marginal} address this by combining the frugal parameterisation with normalising flows to model the marginal interventional distribution directly. Their reliance on copula flows \citep{kamthe2021copula} restricts them to univariate outcomes or static regimes, requires hierarchical copulas for multivariate outputs, and is sensitive to hyperparameters.

Compared to density-estimation generative models, which are restricted either to invertible transformations (frugal flows and RealCause, via normalising flows) or to tractable decoder distributions (CausalEGM, CEVAE, and Credence, via variational autoencoders with Gaussian or Bernoulli decoders), frengression is easier to fit and allows direct control of the causal margin while keeping data simulation tractable.

\subsection{Structure}
We introduce the frengression method in \Cref{sec:method}, the architecture in \Cref{sec:model}, and the estimation procedure and identifiability results in \Cref{sec:estimation}. We provide consistency results for  fitting the model in \Cref{sec:consistency}. Frengression exhibits strength with respect to treatment extrapolation, which is explored in \Cref{sec:extrapolation}. Its flexibility for longitudinal data and survival analysis is demonstrated in \Cref{sec:frengression_seq}. We then provide empirical results to show frengression's comprehensive capability in synthetic experiments (\Cref{sec:synthetic_exp}) and real trial data LEADER (\Cref{sec:leader}). We conclude with potential directions in \Cref{sec:discussion}.

\section{Method}
\label{sec:method}

In this section we describe the necessary notation and basic method for frengression.

\subsection{Notations}
\label{sec:notations}
We use a hat to mark the samples generated from a (fitted) frengression model, for instance, $\widehat{Z}$, $\widehat{X}$, or $\widehat{Y}$. The Euclidean norm of a vector $x\in\R^d$ is denoted $\|x\|$. We use $\P_V$ to denote the distribution of a random variable $V$, $p_V$ for the corresponding density or mass function, $\P^n$ for the empirical measure, and $P(A), A\subseteq \Omega$ to denote the probability of an event $A$. Expectations are represented by $\E$, so $\E f  =\int f d\P$. 
We use $\overline{W}_k$ to denote $(W_0, \ldots, W_{k})$,  the time-varying variable $W$'s values up to the time point $k$. As with the frugal parameterisation (\citealt{evans2024parameterizing}, Remark 1.4), our method fits within any causal inference framework, including \citeauthor{pearl2009causality}'s `do(·)’ operator $P(Y=y\,|\,do(X=x))$. For convenience, in this paper we follow the potential outcomes framework (\citealt{neyman1923applications} and \citealt{rubin1974estimating}), where we write the outcome under a treatment $x \in \mathcal{X}$ as $\Yx$ and its distribution is denoted $\P_{\Yx}$.

In this paper, we consider a general setting depicted in \Cref{fig:general_dag}(a), which describes a simple time-varying confounding scenario. We partition the treatment vector $X = (\Xpaz,\Xchz)$, where $\Xpaz$ denotes the intervention on the treatments causally prior to $Z$, and $\Xchz$ denotes interventions on the treatments causally subsequent to 
$Z$; the respective supports are  $\mathcal{X}_0$ and  $\mathcal{X}_1$. In such a setting, when we write $\Yx$, it is equivalent to $Y\hspace{-.75pt}(x_0,x_1)$. This structure is commonly seen in longitudinal analysis \citep{havercroft2012simulating} and serves as the motivating example for simulating data that satisfies marginal requirements in \citet{evans2024parameterizing}. In the absence of any interventions on the covariates, the causal DAG reduces to the simpler form shown in \Cref{fig:general_dag}(b), which is more typical in static settings.
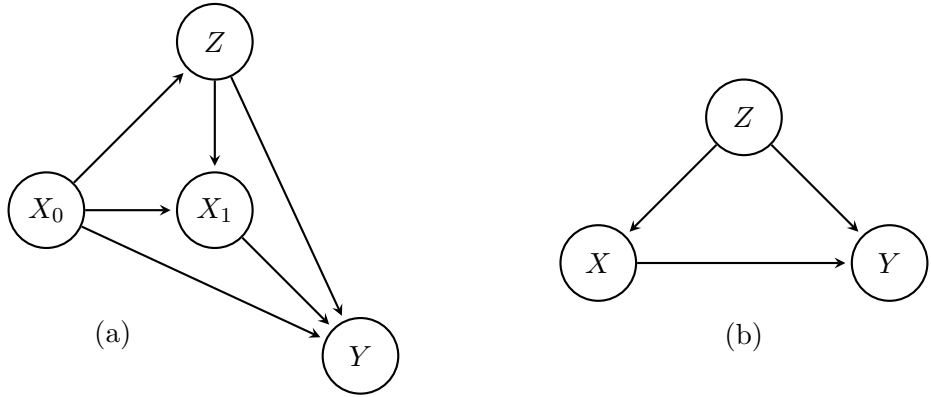
\begin{figure}[ht]
  \centering
  \begin{tikzpicture}[  
      every node/.style = {circle, draw=black, thick,
                           align=center,
                           inner sep=1pt, 
                           minimum size=1cm}, 
      node distance      = 1.2cm and 1.2cm,       
      shorten >= 2pt,                               
      > = stealth                                   
  ]
\begin{scope}

    \node (Z) {$Z$};

    \node[below =of Z]  (X) {$\Xchz$};
    \node [left=of X](XpaZ) {$\Xpaz$};
    \node[below right=of X] (Y) {$Y$};

    \draw[->, thick] (XpaZ) -- (Z);
    \draw[->, thick] (XpaZ) -- (Y);  \draw[->, thick] (XpaZ) -- (X);

    \draw[->, thick] (Z) -- (X);
    \draw[->, thick] (X) -- (Y);
    \draw[->, thick] (Z) -- (Y);
    \node[draw=none, below left of=X, yshift=-8mm, xshift=-5mm] {(a)};
\end{scope}
 

\begin{scope}[xshift=7cm, yshift=-1cm]
    \node(Z) {$Z$};
    \node[below left=of Z]  (X) {$X$};
    \node[below right=of Z] (Y) {$Y$};

    \draw[->, thick] (Z) -- (X);
    \draw[->, thick] (X) -- (Y);
    \draw[->, thick] (Z) -- (Y);
    \node[draw=none, below of=Z, yshift=-17mm] {(b)};
\end{scope}
  \end{tikzpicture}
   \caption{(a) General causal DAG considered in this paper. This time-dependent confounding structure is commonly seen in longitudinal data.  (b) Causal DAG without $\Xpaz$. This structure is more common in time-fixed settings.}
    \label{fig:general_dag}
\end{figure}

\subsection{Review of Key Concepts}
\label{sec:review_of_key_concepts}
 As hinted by its name, frengression employs engression as the generative model. Before we proceed further to more details, we provide a review of concepts related to engression here. 

Parameterised by a neural network, engression aims to estimate the full conditional distribution of a response $Y$ given predictors $W$. This predictive distribution $Y \cmid W=w$ can be taken from a general model class produced by engression and be written as $\mathcal M = \{f(w, \zeta)\}$, where $f: (w,\zeta) \mapsto y$ belongs to a function class and $\zeta$ is a random vector with a pre-specified  distribution, independent of the covariates. A typical choice is that $\zeta$ follows the multivariate standard Gaussian distribution. 
We write $\P_f(y\,|\,w)$ for the distribution of $f(w, \zeta)$, with the randomness coming from $\zeta$, and $\cL(\P,\P_0)$ for the loss function of a distribution $\P$ given a reference $\P_0$. Engression specifically concerns a solution to:
\[
\widetilde{f}=\argmin_{f\in\mathcal{M}} \E_{(Y,W)\sim \P}\,\cL\big(\P_f(y\,|\,W), Y\big).
\]
Although different loss functions can be used in engression (\citealp{shen2024engression}, Appendix D), it shows stable, robust performance using the energy loss $\cL\big({-}\es(\P_f(y\,|\,W), Y)\big)$, where the energy score ($\es$) is defined as 
\begin{equation}
\es (\P,u) = \frac{1}{2}\E_{U,U'\sim \P}\|U-U'\|-\E_{U\sim \P}\|U-u\| \label{eqn:energy_score}
\end{equation}
given a distribution $\P$ and an observation $u$ \citep{gneiting2007strictly}; the $U,U'$ are taken to be independent. The energy score is associated with the energy distance \citep{szekely2013energy}, a distributional divergence that is defined as
\begin{equation}
\label{eqn:energy_distance}
    \ed(\P,\P') = 2\E_{U\sim \P,V\sim \P'}\|U-V\|-\E_{U,U'\sim \P}\|U-U'\|-\E_{V,V'\sim \P'}\|V-V'\|.
\end{equation}
This is a special case of the squared maximum mean discrepancy (MMD)  distance, denoted by $\mmd^2(\P,\P')$, with kernel $k(u,v) = \frac{1}{2}(\|u\|+\|v\|-\|u-v\|)$ \citep{sejdinovic13equivalence}. Using energy loss provides computational simplicity: the estimation of the score or distance is explicitly computed based on sampling. Since the energy distance differs from the energy loss only by a positive scaling plus an additive constant that does not depend on the generator parameters, minimising the energy distance yields the same fitted generator as minimising the energy loss. For compatibility with the relevant literature regarding distributional distance \citep{briol19inference}, we prove the consistency of models trained aiming at minimising energy distance $\ed(\P_f,\P)$ in \Cref{sec:consistency}.

Another speciality of engression is that its model class $\mathcal{M}$ contains the pre-additive noise models (pre-ANMs). Define function classes $\mathcal{F}$ and $\mathcal{G}$. The pre-ANMs can be written as
\begin{equation}
f(w + g(\zeta)) : f \in \mathcal F, g\in\mathcal G,
\end{equation}
where $g(\zeta)$ represents the pre-additive noise; $f$ and $g$ are to be learned. \citet{shen2024engression} provided detailed discussions on the extrapolation properties provided by pre-ANMs, allowing one to  relax the causal effect identification assumptions; details are given in \Cref{sec:extrapolation}.

\subsection{Model and Parametrisation}
\label{sec:model}
\begin{figure}[!tb]
  \centering
  \resizebox{0.9\textwidth}{!}{%
  \begin{tikzpicture}[
      circle node/.style={circle, draw=black, thick, align=center, minimum width=2.5cm, minimum height=2.5cm},
      est node/.style={circle, draw=black, thick, align=center, minimum width=2.5cm, minimum height=2.5cm},
      sim node/.style={circle, draw=black, thick, align=center, minimum width=2.5cm, minimum height=2.5cm},
      >=stealth
  ]
    \node[regular polygon, regular polygon sides=3, draw, minimum size=5cm, fill=gray!30]  {};
    \node[circle node, fill=white] (PZX) at (210:2.8) {
      $\P_\ZX$ \\
      $g(\epsilon)$,\\
      $\epsilon\sim\mathcal{N}(0,I_{d_z+d_x})$
    };
    \node[circle node, fill=white] (phi) at (330:2.8) {
      $\phi^\ast$ \\
      $ h(Z,X,\xi)$,\\
      $\xi\sim\mathcal{N}(0,I_{d_y})$
    };
    \node[circle node, fill=white] (PYX) at (90:2.8) {
      $ \P_{\Yx}$ \\
      $f(x, \eta)$,\\
      $~\eta\sim\mathcal{N}(0,I_{d_y})$
    };

    \node at (0,0) {\large $\P_{ZXY}$};


    \node[est node] (EST) at (-6.5,1) {\textbf{Observation}\\$\{(z,x,y)\}_i$};
    \draw[thick,black,-{Latex[width=3mm,length=4mm]}] (EST.east) -- ++(4cm,0);

    \node[sim node] (SIM) at (6.5,1) {\textbf{Simulation}\\$\{(\hat z,\hat x,\hat y)\}_i$};
    \draw[thick,black,{Latex[width=3mm,length=4mm]}-] (SIM.west) -- ++(-4cm,0);
  \end{tikzpicture}}
  \caption{Estimation and simulation architecture of frengression.}
  \label{fig:causal-architecture}
\end{figure}
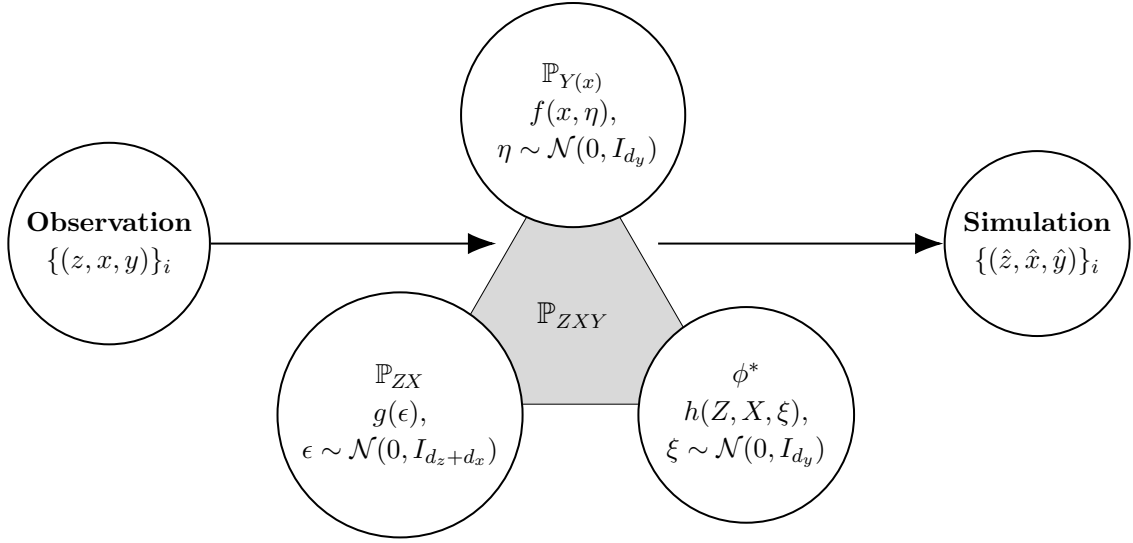

The architecture of frengression is shown in \Cref{fig:causal-architecture}. A frengression model contains three components, each corresponding to one component in the frugal parameterisation:
\begin{itemize}
	\item For the joint distribution $\P_\ZX$, we build a generative model:  
\begin{equation}\label{eq:model_xz}
	g(\epsilon),
\end{equation}
where $\epsilon\sim\cN(0, I_{d_z+d_x})$ and $g:\R^{d_z+d_x}\to\R^{d_z+d_x}$. 

\item Given a specified form of the (marginal) causal effect $\P_{\Yx}$, we use
\begin{equation}\label{eq:model_causal}
	Y=f(x,\eta)
\end{equation}
as the  generative form, where $\eta\sim\cN(0,I_{d_y})$. For example, if $\P_{\Yx}$ is $\cN(x,1)$, it can be specified in a generative form as $f(x,\eta)=x+\eta$.

\item Association remainder $\phi^\ast_\YIZX$. 
Given the causal margin in \eqref{eq:model_causal}, the remaining part to fully specify the distribution  $\P_\YIZX$ is a conditional distribution of $\eta \cmid Z=z, X=x$. For this, we build a conditional generative model
\begin{equation}\label{eq:model_copula}
	\widetilde\eta = h(z,x,\xi),
\end{equation}
where $\xi\sim\cN(0, I_{d_y})$. Similar to the copula measure chosen by the original frugal parameterisation, this term captures the association between $Y$ and $(Z,X)$ in a way that is compatible with the specified causal margin.

\end{itemize}

Frengression parameterises the joint (observed) distribution via generative models, preserving the same properties as the original frugal parametrisation. First, the three components jointly specify the whole distribution $\P_{ZXY}$. To see this, note that model \eqref{eq:model_xz} specifies the distribution for $(Z,X)$,  while models \eqref{eq:model_causal} and \eqref{eq:model_copula} together specify a distribution of $Y\cmid Z,X$. Furthermore, the three components can be chosen to be variation independent. The first component is about the distribution of $(Z,X)$ whereas in the other two components, $X$ and $Z$ are either given or conditioned on; thus, $g$ is variation independent of $(f,h)$. Note that the marginal distribution of $\eta$ is always fixed as the standard Gaussian and is thus not a model component. For the variation independence between $f$ and $h$, note that model \eqref{eq:model_causal} specifies the causal margin $\P_{\Yx}$ only through the function $f$ with a fixed value of $x$ and $\eta$ drawn from its (fixed) marginal distribution, whereas model \eqref{eq:model_copula} is about the conditional distribution of $\eta\cmid Z,X$ with a constraint of a normal margin.
Thus, these two components are variation independent.  More details are given in \Cref{remark:interventional_conditional}. 

We write a frengression model as $\varphi = 
(g,f,h)$. While we specify the distributions of $\epsilon$, $\eta$ and $\xi$ in this paper, their distributions can be non-Gaussian, provided that they are sampled independently of the observed data. Note that even though we draw $\eta$ from a standard normal distribution, because it is subsequently passed through the non-linear function $f(\cdot)$, the interventional distribution induced by frengression need not be Gaussian.

\begin{remark}[Interventional distribution $\P_{\Yx}$ versus conditional distribution $\P_{Y|X=x}$]
\label{remark:interventional_conditional}
    Model \eqref{eq:model_causal} with a fixed $x$ and $\eta$ following its marginal distribution induces the interventional distribution of $\P_{\Yx}$, whereas model \eqref{eq:model_causal} in combination with \eqref{eq:model_copula}, that is $f(X, h(Z,X,\xi))$, induces the conditional distribution of $\P_\YIZX$. Note that in the latter case, $X$ and the noise term $h(Z,X,\xi)$ are correlated, whereas in the `marginal generative causal model' \eqref{eq:model_causal}, $x$ is fixed and so is independent from the noise term $\eta$. 
\end{remark}
\begin{remark}[Target estimand]
The target interventional distribution need not be the fully marginal distribution $\P_{\Yx}$; it can be conditional, $\P_{Y(x)\mid W}$, where $W$ is a chosen set of variables that are held fixed rather than intervened upon. In the static setting of \Cref{fig:general_dag}(b), a common example is $W=Z$.
\end{remark}

\begin{remark}[Privacy]
\label{remark:privacy}
Our generative framework factorises the simulation process into three variation-independent modules, so each can be specified or replaced on its own. In contexts where the variables $Z$ or $X$ contain sensitive information, the model to generate them ($g(\cdot)$) can be replaced by any differentially private generative mechanism. By the post‑processing theorem of differential privacy \citep{dwork2006calibrating}, any further mapping applied to a differentially private algorithm remains differentially private. The entire simulated dataset 
also preserves that privacy guarantee, making our approach naturally extendable to privacy-sensitive settings.
\end{remark}

\subsection{Estimation}
\label{sec:estimation}
As in engression, parameters of each component in frengression are typically estimated by minimising the energy loss, for its simplicity and convenience. However, other loss functions can be used, such as KL divergence.

On the population level, the objective function for model \eqref{eq:model_xz}  is
\begin{equation}
\label{eqn:population_g}
    \cL_\ZX(g) : =\cL(\P_\ZX,  \P_g) =\E\Big[\big\|(Z,X) - g(\epsilon)\big\| - \frac12\big\|g(\epsilon) - g(\epsilon')\big\|\Big].
\end{equation}
The objective for models \eqref{eq:model_causal} and \eqref{eq:model_copula} jointly is
\begin{equation}
\label{eqn:population_fh}
	\cL_\YIZX(f,h):= \E\Big[\big\|Y - f(X,\widetilde\eta)\big\| - \frac12\big\| f(X,\widetilde\eta) -  f(X,\widetilde\eta')\big\|\Big] + \E\Big[\big\|\eta-\bar\eta\big\|-\frac12\big\|\bar\eta - \bar\eta'\big\|\Big],
\end{equation}
where $\widetilde\eta=h(Z,X,\xi)$, $\widetilde\eta'=h(Z,X,\xi')$,  $\bar\eta=h(\bar{Z},\bar{X},\xi'')$, and $\bar\eta'=h(\bar{Z}',\bar{X},\xi''')$ with $\eta\sim \cN(0,I_{d_y})$, $\xi,\xi',\xi'',\xi'''\overset{\rm i.i.d.}{\sim}\cN(0, I_{d_y})$, $(Z,X)\sim\P_\ZX$, $\bar{X}=(\bar X_0,\bar X_1)\sim \P_X$, and, conditional on $\bar X_0$, $\bar{Z},\bar{Z}'\overset{\rm i.i.d.}{\sim}\P_{Z\mid X_0=\bar X_0}$. This is chosen because these two models should be estimated such that:
\begin{itemize}
\item the object $f(X,\widetilde\eta)$, when $\widetilde\eta$ is drawn from its conditional distribution given $(Z,X)$, matches the observed conditional distribution $\P_\YIZX$, which is ensured by the first term in \eqref{eqn:population_fh};
\item the second term in \eqref{eqn:population_fh} calibrates $h$ so that the noise it feeds into $f$ follows the reference distribution $\cN(0,I_{d_y})$; the precise statement, conditional on $X_0$ and for $\P_X$-almost every $x$, is given in \Cref{prop:population_estimates}.
\end{itemize}

As for sampling from $\P_{Z|\Xpaz}$, if $X_0$ does not exist (i.e.~the confounder $Z$ is not affected by any treatment as in \Cref{fig:general_dag}(b)), we do it independently from the marginal distribution $\P_Z$; otherwise, we need to learn an auxiliary model $e^\ast(x_0,\zeta)$ that matches the distribution of $\P_{Z|\Xpaz}$ via the standard engression approach
\begin{equation*}
    e^\ast\in\argmin_e \E\Big[\|Z - e(\Xpaz,\zeta)\| - \frac12\|e(\Xpaz,\zeta) - e (\Xpaz,\zeta')\|\Big],
\end{equation*}
where $\zeta,\zeta'\overset{\rm i.i.d.}{\sim}\cN(0,I_{d_z})$.

We define the population version of frengression, denoted by $\varphi^\ast=(g^\ast,f^\ast,h^\ast)$, as the solution to minimising the objective functions:
\begin{align*}
    g^\ast &= \argmin_{g}\cL_\ZX(g)\\
    (f^\ast,h^\ast) &= \argmin_{f,h}\cL_\YIZX(f,h).
\end{align*}

To show the population guarantees that frengression learns each component correctly and recovers the true data generating process of $(Z,X,Y)$, we impose the following assumptions:

\begin{assumption}[Correct specification]
\label{ass:correct_model}
We write the function classes $\mathcal{G} =\{g(\epsilon)\}$, $\mathcal{F}=\{f(x,\eta)\}$, and $\mathcal{H}=\{h(z,x,\xi)\}$. There exist $g'\in \mathcal{G}$, $f'\in\mathcal{F}$, and $h'\in \mathcal{H}$ such that
\begin{enumerate}[label=(\roman*)]
    \item $g'(\epsilon)\sim \P_\ZX$ with $\epsilon\sim \cN(0,I_{d_x+d_z})$;
    \item $f'\big(x,h'(z,x,\xi)\big)\sim \P_{Y|Z=z,X=x}$ for $\P_\ZX$-almost every $(z,x)$, with $\xi\sim\cN(0,I_{d_y})$;
    \item for $\P_X$-almost every $x=(x_0,x_1)$, $h'(Z_x,x,\xi)\sim \cN(0,I_{d_y})$ when $Z_x\sim \P_{Z\mid X_0=x_0}$ and $\xi\sim\cN(0,I_{d_y})$.
\end{enumerate}
\end{assumption}

\begin{assumption}[Identification assumptions regarding DAG in \Cref{fig:general_dag}(a)]\label{ass:identification} 
\,
\begin{enumerate}[label=(\roman*)]
    %
  \item \textbf{Positivity:} The joint support of $(\Xpaz,Z,\Xchz)$ is full, that is,
\begin{align*}
    p_{\XpZXc}(x_0, z, x_1) > 0
\end{align*}
for almost all $(x_0,z,x_1)\in \mathcal{X}_0 \times \mathcal{Z}\times \mathcal{X}_1$. 
  \label{ass:positivity}
      \item \textbf{Unconfoundedness:} 
      The potential outcome $\Yx$  is independent of the observed treatment given past covariates and treatment:
      \[
            \Yx  \indep  \Xchz \bigm| Z=z,\Xpaz=x_0, \quad\text{and}\quad
             \Yx\indep  X_0,
              \quad \forall\,x\in\mathcal X.
            \]
  \item \textbf{Consistency:}
    If $X = x$ then the observed outcome equals the potential outcome at that level:
    $$
      X = x \quad \mbox{implies} \quad Y = \Yx, 
      \quad \forall\,x\in\mathcal X.
    $$
\end{enumerate}
\end{assumption}

When the conditions in \Cref{ass:identification} hold, the potential outcome distribution $\P_{\Yx}$ is identifiable from the observed joint distribution of $(Z,X,Y)$ in \Cref{fig:general_dag}(a):
\begin{lemma}[Identifiability]
    When \Cref{ass:identification} is satisfied, we have:
    \[
    p_{\Yx}(y) = \int_{\mathcal{Z}} p_{Y|\Xpaz Z \Xchz}(y\cmid x_0,z,x_1) \, p_{Z|\Xpaz}(z\cmid x_0) \, \mathrm{d} z.
    \]
\label{lem:identifiability}
\end{lemma}
The proof can be found in \Cref{sec:proof_identifiability}. In \Cref{sec:extrapolation} we show that a pre-additive noise structure extrapolates the median of the causal margin beyond the observed treatment support, relaxing positivity for this functional.

\Cref{ass:correct_model} asserts the existence of functions $(g',f',h')$, while \Cref{ass:identification} guarantees unique potential outcome distributions. With these assumptions, we introduce the following proposition, which provides the population guarantee that frengression correctly learns the true distributions when the model is correctly specified.

\begin{proposition}[Correctness]
\label{prop:population_estimates}
When Assumptions \ref{ass:correct_model} and \ref{ass:identification} hold, the frengression solution $\varphi^\ast=(g^\ast,f^\ast,h^\ast)$ satisfies, for $\P_\ZX$-almost every $(z,x)$ and $\P_X$-almost every $x$,
    \begin{align*}
        g^\ast(\epsilon)&\sim \P_\ZX,\\
        f^\ast(x, h^\ast(z,x,\xi)) &\sim \P_{Y|Z=z,X=x},\\
        f^\ast(x,\eta) &\sim \P_{\Yx},
    \end{align*}
where $\epsilon\sim \cN(0,I_{d_z+d_x})$, $\eta \sim \cN(0,I_{d_y})$, and $\xi\sim\cN(0,I_{d_y})$.
\end{proposition}

The proof of \Cref{prop:population_estimates} is provided in \Cref{sec:proof_population_estimates}.
\subsubsection{Finite-sample Frengression}
\label{sec:finite_frengression}

$\varphi^\ast$ solves the population objectives in (\ref{eqn:population_g}) and (\ref{eqn:population_fh}). In practice, we only have access to finite samples $\{(Z_i, X_i, Y_i)\}_{i=1}^n \sim \P_{ZXY}(z,x,y)$; the empirical frengression estimator $\hat\varphi=(\hat g,\hat f,\hat h)$ is defined by
\begin{align*}
    \hat{g} &= \argmin_g\widehat{\cL}_\ZX(g)\\
    (\hat{f},\hat{h}) &= \argmin_{f,h}\widehat{\cL}_\YIZX(f,h),
\end{align*} 
where the empirical loss function estimates are
\small
\begin{align}
   \widehat{\cL}_\ZX(g) &:= \frac{1}{n}\sum_{i=1}^n\bigg[\frac1m\sum_{j=1}^m\big\|(X_i,Z_i) - g(\epsilon_{i,j})\big\| - \frac{1}{2m(m-1)}\sum_{j=1}^m\sum_{\substack{j'=1 \\ j'\neq j}}^m \big\|g(\epsilon_{i,j}) - g(\epsilon_{i,j'})\big\|\bigg],
    \label{eq:empirical_g}
\end{align}
\begin{equation}\label{eq:empirical_fh}
\begin{split}
  \widehat{\cL}_\YIZX(f,h)&:=\frac{1}{n}\sum_{i=1}^n\bigg[\frac{1}{m}\sum_{j=1}^m\big\|Y_i - f(X_i,\widetilde\eta_{i,j
    })\big\|  - \frac{1}{2m(m-1)}\sum_{j=1}^m\sum_{\substack{k=1\\k\neq j}}^m\big\| f(X_i,\widetilde\eta_{i,j}) - f(X_i,\widetilde\eta_{i,k})\big\|\bigg] \\[.5ex]
  &\qquad\qquad\mathbin{+}\frac{1}{n}\sum_{i=1}^{n} \bigg[\frac{1}{m}\sum_{j=1}^{m}\big\|\eta_i-\bar\eta_{i,j}\big\|-\frac{1}{2m(m-1)}\sum_{j=1}^m\sum_{\substack{k=1\\k\neq j}}^m\big\|\bar \eta_{i,j}- \bar\eta_{i,k}\big\|\bigg],
\end{split}
\end{equation}
\normalsize
respectively. Here, for $i=1,\ldots,n$ and $l=1,\ldots,m$, we simulate $\epsilon_{i,l} \sim \cN(0, I_{d_z+d_x})$, $\eta_{i}\sim \cN(0,I_{d_y})$ and $\xi_{i,l} \sim \cN(0, I_{d_y})$, and set $\widetilde\eta_{i,l} = h(Z_i,X_i,\xi_{i,l})$ and $\bar{\eta}_{i,l}=h(Z'_{i,l}, X_i,\xi_{i,l})$. The draws $Z'_{i,l}$ follow the same conditional distribution as $\bar Z$ in \eqref{eqn:population_fh}: if $X_0$ is absent (the common time-fixed setting of \Cref{fig:general_dag}(b)), no auxiliary model is needed and $Z'_{i,l}$ is taken from a random permutation of the observed $Z$'s, which estimates $\P_Z$; only when $X_0$ is present are the draws taken from the auxiliary model, $Z'_{i,l}=e(X_{0i},\zeta_{i,l})$, estimating $\P_{Z\mid X_0=X_{0i}}$.
Only in the latter case, the auxiliary model $e$ for $\P_{Z\mid X_0}$ is fitted by minimising
\small
\begin{align}
     \widehat{\cL}_{Z|X_0}(e) &:= \frac{1}{n}\sum_{i=1}^n\bigg[\frac1m\sum_{j=1}^m\big\|Z_i - e(X_{0i}, \zeta_{i,j})\big\| - \frac{1}{2m(m-1)}\sum_{j=1}^m\sum_{\substack{j'=1\\j'\neq j}}^m \big\|e(X_{0i}, \zeta_{i,j}) - e(X_{0i}, \zeta_{i,j'})\big\|\bigg],\label{eq:empirical_e}
\end{align}
\normalsize 
where $\zeta_{i,l}\sim\mathcal{N}(0,I_{d_z})$, for $i=1,\ldots,n$, $l=1,\ldots,m$.

The empirical losses \eqref{eq:empirical_g} and \eqref{eq:empirical_fh} estimate the corresponding population objectives; the second term of \eqref{eq:empirical_fh} is unbiased when the auxiliary draws are taken from the true conditional distribution $\P_{Z\mid X_0}$, which \eqref{eq:empirical_e} estimates when $X_0$ is present. In \Cref{sec:consistency}, we present theoretical guarantees that, by minimising the finite-sample empirical losses, the fitted generative distribution converges to the distribution within the model class that is closest to the target in energy distance.

In practice, we parameterise $g$, $f$ and $h$ with neural networks. The resulting empirical objectives are therefore almost everywhere differentiable with respect to every model parameter, allowing us to optimise them with stochastic gradient descent to obtain the fitted frengression model.

\subsection{Sampling}
\label{sec:sampling_static}
Given the generative nature of our model, sampling from a fitted frengression is straightforward. We provide the sampling procedures from the population version $\varphi^\ast$  for four common simulation scenarios below:
\begin{enumerate}
    \item Simulate $(\widehat{Z},\widehat{X})$.\\
          Sample $\epsilon \sim \mathcal{N}\!\big(0, I_{d_z+d_x}\big)$ and set
              $(\widehat{Z},\widehat{X}) = g^\ast(\epsilon)$.

    \item Simulate $\widehat{Y}(x)$.\\
          For a treatment value $x$,  
          draw $\eta \sim \mathcal{N}\!\big(0, I_{d_y}\big)$ and compute
              $\widehat{Y} (x)= f^\ast\big(x, \eta\big)$.

    \item Simulate $(\widehat{Z},\widehat{X},\widehat{Y})$ from the joint distribution of $(Z,X,Y)$.
          \begin{enumerate}[label=(\roman*)]
              \item Generate $(\widehat{Z},\widehat{X})$ as in step 1. 
              \item Draw $\xi \sim \mathcal{N}\!\big(0, I_{d_y}\big)$ and obtain
                        $\widetilde{\eta} = h^\ast\big(\widehat{Z},\widehat{X},\xi\big)$.
              \item Produce the response via
                        $\widehat{Y} = f^\ast\big(\widehat{X}, \widetilde{\eta}\big)$.
          \end{enumerate}
    \item Simulate observed $(\widehat{Z},\widehat{X})$ together with a potential outcome $\widehat{Y}(x')$ \citep[as in single world intervention graphs,][]{richardson2013single}:
    \begin{enumerate}[label=(\roman*)]
        \item Generate $(\widehat{Z},\widehat{X})$ and obtain $\widetilde{\eta} = h^\ast\big(\widehat{Z},\widehat{X},\xi\big)$.
        \item Produce the response via $f^\ast(x',\widetilde\eta)$ with a specified $x'$.

    \end{enumerate}
\end{enumerate}

The $(g^\ast,f^\ast,h^\ast)$ can be replaced by model $(\hat{g},\hat{f},\hat{h})$ fitted using finite samples. Taken together, these sampling schemes allow us to generate synthetic observations from any component or the full joint distribution of a model fitted on real data.


\subsection{Sampling-based Inference}
Frengression allows sampling-based inference. For instance, a typical quantity of interest is the average potential outcome $\E \Yx$ under a specific intervention $x$. A fitted frengression model contains $\hat{f}(x,\eta)$, and this estimates the marginal outcome distribution. Thus, for any given intervention value $x$, we sample $\eta_i\sim \cN(0,I_{d_y})$ for $i=1,\ldots,m$. Then we obtain a sample $\{Y_i=\hat{f}(x,\eta_i),i=1,\ldots,m\}$ from the estimated distribution $\hat{\P}_{\Yx}$. The expected potential outcome  can thus be estimated using the empirical mean $\frac{1}{m}\sum_{i=1}^m \hat{f}(x,\eta_i)$. Besides the interventional mean, any other quantity can also be estimated; for instance, the $\alpha$-quantile 
of $\P_{\Yx}$ is consistently estimated by the $\alpha$-quantile of this sample.


\section{Model Consistency}
\label{sec:consistency}

In this section we first provide a general consistency result for generative modelling with the empirical energy distance as the loss function, and then adapt it to the components of frengression. The generative model $\textsl{g}_\theta$ (distinct from the component generator $g$ of frengression), parameterised by $\theta\in\Theta$, takes a random noise input $\nu$, independent of the data and drawn from a fixed reference distribution; we assume it is a standard Gaussian, though any other distribution under which the assumptions below hold may be used instead. The model returns the sample $\textsl{g}_\theta(\nu)$, and we write $\P_\theta$ for its distribution, the distribution it generates. The aim is to learn an unknown data-generating distribution $\Q$, and since frengression is parameterised by a neural network, $\theta$ collects the network weights.

In this paper, we train the generative model $\textsl{g}_\theta$ by minimising the energy distance between $\P_\theta$ and $\Q$. Following the definition in \eqref{eqn:energy_distance}, we write the population energy distance as 
\begin{equation}
    \ed(\P_\theta,\Q) = 2\E_{X\sim \P_\theta,Y\sim \Q}\|X-Y\|-\E_{X,X'\sim \P_\theta}\|X-X'\|-\E_{Y,Y'\sim \Q}\|Y-Y'\|.
\end{equation} 

We make the following assumption:

\begin{assumption}[Existence of $\theta^\ast$]
\label{ass:existence}
We assume there exists a minimiser $$\theta^\ast = \arg\min_{\theta\in\Theta}\ed(\P_\theta,\Q).$$
\end{assumption}
$\theta^\ast$ is what we would like to estimate if it exists, as $\P_{\theta^\ast} \in  \{\P_\theta, \theta\in\Theta\}$ is the `closest' distribution to $\Q$ generated by $\textsl{g}_\theta$.

In practice, it is not common to have direct access to $\Q$; instead, we generally have access to empirical observations from $\Q$: $\{y_1,\ldots,y_n\}$, forming the empirical measure $\Q^n = \frac{1}{n}\sum_{i=1}^n\delta_{y_i}$. When $\textsl{g}_\theta$ is analytically accessible, an estimator could be found as
\begin{align*}
    \hat\theta_n = \arg\min_{\theta\in\Theta}\ed(\P_\theta,\Q^n),
\end{align*} 
where
\[
\ed(\P_\theta,\Q^n) = \frac{2}{n}\sum_{j=1}^n \E_{X\sim \P_\theta}\|X-y_j\| - \E_{X, X'\sim \P_\theta}\|X-X'\|-\frac{1}{n(n-1)}\sum_{j=1}^n\sum_{\substack{j'=1\\j\neq j'}}^n\|y_j-y_{j'}\|.
\]
If $\textsl{g}_\theta$ is not directly accessible, but samples can be independently drawn from it, the estimator then becomes
\begin{align*}
    \hat\theta_{m,n} = \arg\min_{\theta\in\Theta}\ed(\P^m_\theta,\Q^n),
\end{align*}
where 
\[
\ed(\P^m_\theta,\Q^n) = \frac{2}{nm}\sum_{i=1}^m\sum_{j=1}^n \|x_i-y_j\| - \frac{1}{m(m-1)}\sum_{i=1}^m\sum_{\substack{i'=1\\i'\neq i}}^m\|x_i-x_{i'}\|-\frac{1}{n(n-1)}\sum_{j=1}^n\sum_{\substack{j'=1\\j\neq j'}}^n\|y_j-y_{j'}\|.
\] 
$\ed(\P^m_\theta,\Q^n)$ aligns with the finite-sample objective functions defined in \Cref{sec:finite_frengression}.

Under the assumption that the minimiser $\theta^\ast = \argmin_{\theta\in\Theta}\mmd(\P_\theta,\Q)$ exists, \citet{briol19inference} proved the consistency of  $\hat\theta_{n}$ and $\hat\theta_{m,n}$, i.e.~$\lim_{n\to\infty}\hat\theta_{n}=\theta^\ast$ and $\lim_{m,n\to\infty}\hat\theta_{m,n}=\theta^\ast$,  when the kernel used in the MMD is bounded. However,  frengression is trained using energy distance, which is a special case of MMD employing an unbounded, Euclidean distance-based kernel. Although it is reasonable to argue that a bounded covariate space may lead to a bounded kernel in practice, our goal is to extend the theoretical results of \citeauthor{briol19inference} to cover the unbounded kernels used in energy distance, under tail conditions on the data. Here we provide the consistency result under the following alternative two assumptions.

\begin{assumption}[Sub-exponential tails]
\label{ass:tail}
There exist constants $C,\kappa>0$ such that, for every probability measure $\P$
appearing in the analysis, a random variable $X\sim \P$ satisfies
$$
\Pr(\|X\|>t) \leq C\exp(-\kappa t), \qquad\forall t>0.
$$
\end{assumption}

We first present the convergence result based on energy distance:
\begin{theorem}
\label{thm:convergence}
    Suppose:
    \begin{enumerate}
        \item[(a)] \Cref{ass:existence} and \Cref{ass:tail} hold;
        \item[(b)] the model is correctly specified, that is, $\P_{\theta^\ast}\overset{d}{=}\Q$.
    \end{enumerate}
    Then
    \begin{align*}
        \ed \big\{\P_{\hat{\theta}_n},\Q\big\} \to 0 \quad\text{in probability,}\qquad\text{and hence}\qquad \P_{\hat{\theta}_n}\Rightarrow\Q \quad\text{in probability.}
    \end{align*}
    If, in addition,
    \begin{enumerate}
        \item[(c)] $\Theta$ is compact, and
        \item[(d)] for each $\nu$, the map $\theta\mapsto \textsl{g}_\theta(\nu)$ is Lipschitz on $\Theta$, that is, $\|\textsl{g}_\theta(\nu)-\textsl{g}_{\theta'}(\nu)\|\le L(\nu)\,\|\theta-\theta'\|$ for all $\theta,\theta'\in\Theta$, with  $L(\nu)$ satisfying $\E[\sqrt{L(\nu)}]<\infty$,
    \end{enumerate}
    then the same conclusions hold for the finite-sample estimator $\hat\theta_{m,n}$ as $m,n\to\infty$.
\end{theorem}

Applied to the integrated objectives \eqref{eq:empirical_g} and \eqref{eq:empirical_fh}, \Cref{thm:convergence} yields distributional consistency for each frengression component: under conditions (a)--(d), the fitted distributions converge weakly to those induced by the population minimisers, which by \Cref{prop:population_estimates} coincide with the true conditional and causal-margin distributions.

\begin{corollary}[Convergence of frengression components]
\label{cor:component_convergence}
Let $(\hat g, \hat f, \hat h)$ be the frengression estimators obtained by
minimising the empirical objectives \eqref{eq:empirical_g} and
\eqref{eq:empirical_fh}. Suppose Assumptions~\ref{ass:correct_model} and \ref{ass:identification} hold, together with conditions (a)--(d) of \Cref{thm:convergence}.
Then, as $m,n\to\infty$, for $\P_{ZX}$-almost every $(z,x)$ and $\P_X$-almost every $x$,
\begin{align*}
\hat g(\epsilon)               &\;\Rightarrow\; \P_{ZX}, \\
\hat f\big(x,\hat h(z,x,\xi)\big) &\;\Rightarrow\; \P_{Y\mid Z=z,X=x}, \\
\hat f(x,\eta)                 &\;\Rightarrow\; \P_{\Yx},
\end{align*}
where $\epsilon\sim\mathcal N(0,I_{d_z+d_x})$, $\eta\sim\mathcal N(0,I_{d_y})$
and $\xi\sim\mathcal N(0,I_{d_y})$ are drawn independently of the data.
\end{corollary}

The proof can be found in \Cref{sec:proof_consistency}.

Because the model is parameterised by a neural network, the minimiser $\theta^\ast$ need not be unique: distinct weights can induce the same $\P_\theta$. We therefore state consistency at the level of distributions, as in \Cref{thm:convergence}, rather than for the parameter $\theta$ itself. Correct specification enters only to identify the limit as $\Q$; without it, the same argument gives
$$
\ed(\P_{\hat\theta_n},\Q)-\inf_{\theta\in\Theta}\ed(\P_\theta,\Q)\to0\quad\text{in probability},
$$
so a fitted frengression attains the smallest energy distance to $\Q$ achievable within the model class. The proof is provided in \Cref{sec:proof_consistency}.

\section{Extrapolation}
\label{sec:extrapolation}
For simplicity, in this section, we consider the case with univariate treatment $X$ and outcome $Y$ but allow for multiple covariates $Z$.
\begin{assumption}[Pre-ANM]\label{ass:preanm}
    The causal margin can be parametrised via a pre-additive noise model 
    \begin{equation*}
        Y = \widetilde{f}(X+\widetilde\eta),
    \end{equation*}
    where $\widetilde{f}$ is any strictly monotone and continuously differentiable function and $\widetilde\eta$ has a continuous, positive density on its support and is independent of $X$. Without loss of generality, assume $\widetilde\eta$ has a median 0 and is supported on $[-\widetilde\eta_0,\widetilde\eta_0]$ for some $\widetilde\eta_0>0$; we let $\widetilde\eta_0=\infty$ when $\widetilde\eta$ has an unbounded support. 
\end{assumption}
This assumption holds, for example, when the structural equation of $Y$ is given by $Y=\widetilde{f}(X+\widetilde{h}(Z,\widetilde\xi))$, where $\widetilde{h}$ is any bivariate function taking both $Z$ and $\widetilde\xi$ as arguments, the noise $\widetilde\xi\sim\mathcal{N}(0,1)$ can be correlated with $X$, and when the treatment does not causally affect the covariates. 

We study the extrapolation performance of frengression for median estimation of the causal margin. Let $\widetilde{m}(x)$ denote the true median of $\Yx$ for any $x$. When Assumption~\ref{ass:preanm} holds, the frengression model class can be smaller for Assumption~\ref{ass:correct_model} to hold. Specifically, we reparametrise the model for the causal margin as $f(X+f_1(\eta))$, where $f$ and $f_1$ are both strictly monotone and continuously differentiable functions, with the normalisation $f_1(0)=0$, without loss of generality \citep{shen2024engression}. Denote the median of the population frengression solution as $m^\ast(x)$, which is defined as the median of $f^\ast(x,\eta)$.

The following proposition indicates that frengression can recover the true median of causal margin for $x$ outside of the support. The proof can be found in \Cref{sec:proof_extrap}.
\begin{proposition}[Extrapolation]\label{prop:extrap}
    Let the support of $X$ be $\cX=[x_{\rm min}, x_{\rm max}]$. Under Assumption~\ref{ass:preanm}, we have 
    \begin{equation*}
        m^\ast(x) = \widetilde{m}(x)
    \end{equation*}
    for all $x\in[x_{\rm min}-\widetilde\eta_0, x_{\rm max}+\widetilde\eta_0]$.
\end{proposition}



The experimental results provided in \Cref{fig:adrf_estimate} show that frengression maintains accurate  estimates even when extrapolating beyond the observed values of treatment $X$.

\section{Frengression on Sequential Data}
\label{sec:frengression_seq}

\begin{figure}[!tb]
  \centering
\begin{tikzpicture}[
  global/.style={circle, draw=orange!80, 
  thick, minimum size=10mm, inner sep=0.5mm},
  covariate/.style={circle, draw=green!50, 
  thick, minimum size=10mm, inner sep=0.5mm},
  treatment/.style={circle, draw=cyan!80, 
  thick, minimum size=10mm, inner sep=0.5mm},
  outcome/.style={circle, draw=black, fill=white, thick, minimum size=10mm, inner sep=0.5mm},
  edgeC/.style={-Stealth, orange!80, thick},
  edgeZ/.style={-Stealth, green!50, thick},
  edgeX/.style={-Stealth, cyan!80, thick},
  edgeY/.style={-Stealth, black, thick},
  node distance=1.2cm and 2.4cm, 
]

\node[global]           (C)    {$C$};

\node[covariate]           (Z0)   [right=of C,    yshift=+7mm] {$\overline{Z}_{t-1}$};
\node[covariate]           (Z1)   [right=of Z0]             {$Z_t$};

\node[treatment]         (X0)   [below=of Z0]             {$\overline{X}_{t-1}$};
\node[treatment]         (X1)   [below=of Z1]             {$X_t$};

\node[outcome]         (Y0)   [below=of X0, xshift = +16mm]             {$\overline{Y}_{t-1}$};
\node[outcome]         (Y1)   [below=of X1, xshift = +16mm]             {$Y_t$};


\draw[edgeC] (C) -- (X0);
\draw[edgeC] (C) -- (X1);

\draw[edgeC, bend right=40] (C) to (Y0);
\draw[edgeC, bend right=55] (C) to (Y1);

\draw[edgeC, bend left=35] (C) to (Z0);
\draw[edgeC, bend left=40] (C) to (Z1);

\foreach \i in {0,1} {
  \draw[edgeZ] (Z\i) -- (X\i);
}
\draw[edgeZ] (Z0) -- (X1);

\foreach \i in {0,1} {
  \draw[edgeX] (X\i) -- (Y\i);
}
\draw[edgeX] (X0) -- (Y1);

\draw[edgeZ] (Z0) -- (Z1);

\draw[edgeX] (X0) -- (Z1);

\draw[edgeX] (X0) -- (X1);

\foreach \i in {0,1} {
  \draw[edgeZ] (Z\i) -- (Y\i);
}
\draw[edgeZ] (Z0) -- (Y1);




\end{tikzpicture}
\caption{Example diagram of longitudinal data. For simplicity we present one time point $t$ as an example, and we show that $(Z_t, X_t, Y_t)$ are connected with all previous time steps. As with the static case, there are no dimensional constraints on the variables.}
\label{fig:longitudinal_demonstration}
\end{figure}
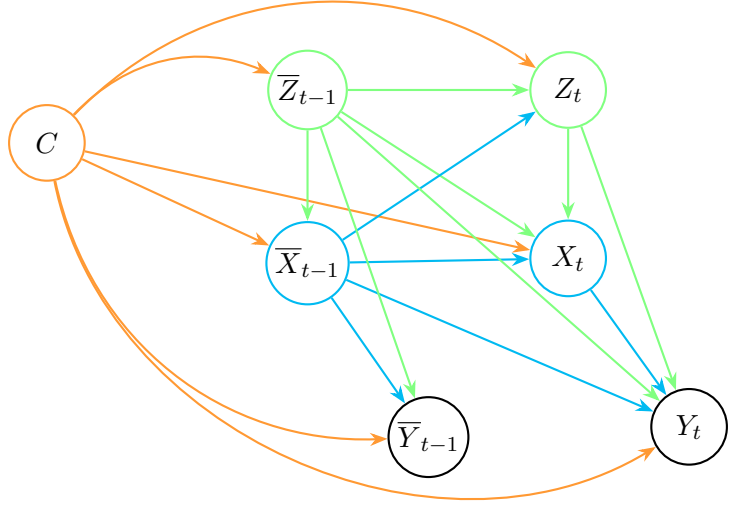

The discussion in Sections \ref{sec:method} and \ref{sec:consistency} introduced the basic frengression machinery in the static case. Many scientific questions, however, involve time-varying decision points, such as longitudinal studies and survival analysis, where treatment decisions may change over time in response to evolving patient histories. 
The frugal parameterisation is particularly helpful in longitudinal analysis, as the time-varying covariates are marginalised out. Leveraging this idea, especially the recursive/nested variant of frugal parameterisation introduced in Section 6 of \citet{evans2024parameterizing}, frengression extends naturally to the longitudinal context.

An illustrative causal diagram is given in \Cref{fig:longitudinal_demonstration}.  
Let $C\in\mathbb{R}^{d_c}$ denote the static, pre-treatment baseline covariates; for the ease of notation, we use $Z$, $X$, $Y$ to denote the vector: 
$Z=(Z_0,\ldots,Z_{T-1})$ with $Z_t\in\mathbb{R}^{d_z}$ the time-varying covariates;  
$X=(X_0,\ldots,X_{T-1})$, with $X_t\in\mathbb{R}^{d_x}$ the sequence of interventions; and  
$Y=(Y_0,\ldots,Y_{T-1})$, $Y_t\in\R^{d_y}$ the observed outcomes.  In the special case of survival-valued outcomes, $Y_t\in\{0,1\}$ with $Y_t=1$ indicating an event occurrence; simulation stops after such an event.  
Other than survival variables, all can assume continuous or discrete values. 

\subsection{Model}
We illustrate the following sections with an example where the interventional distribution of interest is $\P_{\Yx|C}$. 

For data as described above, we write the sequential frengression as $\varphi_s$. We first build a model to generate the baseline covariates $C$ using the unconditional engression $g_c(\varsigma)$, where $\varsigma\sim \cN(0, I_{d_c})$. Then we build the sequential frengression $(\varphi_0,\ldots, \varphi_{T-1})$, where $\varphi_t = (g_t,f_t, h_t)$ for $t=0,\ldots,T-1$. For each $t \in \{0,\ldots,T-1\}$:
\begin{itemize}
    \item When $t=0$, build $g_0(c, \epsilon)$ for the joint distribution $\P_{Z_0 X_0|C}$; otherwise $g_t(c,\overline{z}_{t-1}, \overline{x}_{t-1}, \epsilon)$ for the joint distribution $\P_{Z_t X_t|C \overline{Z}_{t-1} \overline{X}_{t-1}}$. In both cases, $\epsilon \sim \cN(0, I_{d_z+d_x})$.
    \item $Y=f_t(c,\overline{x}_t,\eta)$, $\eta\sim \cN(0,I_{d_y})$ for the interventional distribution $
    \P_{Y_t(\overline{x}_t)|C}$ of treatment values up to time $t$ given baseline covariates $C$.
    \item $\widetilde\eta = h_t(c,\overline{z}_{t},\overline{x}_{t},\xi)$ where $\xi \sim \cN(0,I_{d_y})$; this term is used for modelling the association.
\end{itemize}
The sequential frengression consists of all the components above: $\varphi_s=(g_c(\varsigma),\varphi_0,\varphi_1,\ldots,\varphi_{T-1})$. For convenience, we also denote $g=(g_0, \ldots, g_{T-1})$, $f=(f_0, \ldots, f_{T-1})$ and $h=(h_0, \ldots, h_{T-1})$, then $\varphi_s$ can also be expressed as $\varphi_s=(g_c,g,f,h)$.

For longitudinal settings we introduce two variants of frengression, \textsl{FrengressionSeq} and \textsl{FrengressionSurv}.  
Both inherit the modelling architecture described above; the distinction is that \textsl{FrengressionSurv} is tailored to survival data, where the effective sample size shrinks over time as events occur.  
Accordingly, each $\varphi_t$ is trained on the subpopulation that is still at risk at time $t$.
This renormalisation does not affect the generative procedure, because at simulation time we censor a synthetic trajectory as soon as an event occurs and remove it from all subsequent steps, ensuring that the simulated data follow the same at-risk pattern as the training cohort.  
These sequential variants extend frengression to right-censored settings while preserving its generative flexibility.

\subsection{Estimation}

The population objective for $g_c$ is the analogue of \eqref{eqn:population_g} with $(Z_i,X_i)$ replaced by $C_i$. Estimation objectives of $g, f, h$ are the same as in \eqref{eqn:population_g} and \eqref{eqn:population_fh} respectively, only now the outcome variable contains $T$ time steps of dimension $d_y\times T$, and here we specify causal estimand of interest to be conditional on baseline covariates, which can be customised according to the specific question of interest. More formally, the population solution to $\varphi_s$  can be expressed as:
\begin{align*}
    g^\ast_c&:=\argmin_{g_c}\E\Big[\big\|C - g_c(\epsilon_c)\big\| - \frac12\big\|g_c(\epsilon_c) - g_c(\epsilon_c')\big\|\Big]\\
   g^\ast&:=(g_0^\ast,\ldots,g_{T-1}^\ast)=\argmin_{g}\E\Big[\big\|(Z,X) - g(C,\epsilon)\big\| - \frac12\big\|g(C,\epsilon) - g(C,\epsilon')\big\|\Big]\\
   	(f^\ast,h^\ast)&:= (f_0^\ast,\ldots, f_{T-1}^\ast, h_0^\ast,\ldots, h_{T-1}^\ast)\\
    &=\argmin_{f,h}\E\Big[\big\|Y - f(C,X,\widetilde\eta)\big\| - \frac12\big\| f(C,X,\widetilde\eta) -  f(C,X,\widetilde\eta')\big\|\Big] 
    + \E\Big[\big\|\eta-\bar\eta\big\|-\frac12\big\|\bar\eta - \bar\eta'\big\|\Big]
\end{align*}
Here $g(C,\epsilon)$ denotes the sequential generator obtained by composing $g_0,\ldots,g_{T-1}$ recursively over time, where $\epsilon_c,\epsilon_c' \overset{\rm i.i.d.}{\sim} \cN(0,I_{d_c})$; $\epsilon,\epsilon'\overset{i.i.d}{\sim}\cN(0,I_{(d_z+d_x)\times T})$; $\widetilde\eta=h(C,Z,X,\xi)$, $\widetilde\eta'=h(C,Z,X,\xi')$,  $\bar\eta=h(C,\bar{Z},\bar{X},\xi'')$, and $\bar\eta'=h(C,\bar{Z}',\bar{X},\xi''')$ with $\eta\sim \cN(0,I_{d_y\times T})$, $\xi,\xi',\xi'',\xi'''\overset{\rm i.i.d.}{\sim}\cN(0, I_{d_y\times T})$, $(Z,X)\sim\P_{ZX|C}$, $\bar{X}\sim \P_X$ and $\bar{Z},\bar{Z}'\overset{\rm i.i.d.}{\sim}\P_{Z|C\overline{X}_{T-2}}$. Similar to that in the static setting, we learn auxiliary models $e:=(e_0,e_1,\ldots e_{T-1})$ to match $\P_{Z|C\overline{X}_{T-2}}$ via standard engression models at each time step $t$:
\begin{align*}
    e_0^\ast &:= \argmin_{e_0} \E\Big[\|Z_0 - e_0(C,\zeta_0)\| - \frac12\|e_0(C,\zeta_0) - e_0 (C,\zeta_0')\|\Big] \\
    \text{and} \quad e_t^\ast &:= \argmin_{e_t} \E\Big[\|Z_t - e_t(C,\overline{X}_{t-1}, \zeta_t)\| - \frac12\|e_t(C,\overline{X}_{t-1},\zeta_t) - e_t (C,\overline{X}_{t-1},\zeta_t')\|\Big],
\end{align*}
where $\zeta_0,\zeta_0', \zeta_t,\zeta_t' \overset{\rm i.i.d.}{\sim} \mathcal{N}(0,I_{d_{z}})$, for $t\in\{1,\ldots,T-1\}$.

Because treatments vary over time, we impose the standard \emph{sequential conditional exchangeability} assumption to ensure
identification.  For detailed discussions, see \citet{miguel2023causal} and
\citet[Remark 1.5; Sections 6–7]{evans2024parameterizing}. With this assumption in place, the population guarantees established for static frengression carry over to the sequential model at every time step.

\begin{corollary}[Sequential model preserves correctness]
\label{cor:frengressionseq_correctness}
If at each time point the static frengression model satisfies the correctness property (\Cref{prop:population_estimates}), then the joint sequential model also satisfies correctness.
\end{corollary}

\begin{corollary}[Sequential model preserves consistency]
\label{cor:frengressionseq_consistency}
If at each time point the static frengression model enjoys distributional consistency (\Cref{thm:convergence}), then the joint sequential model is likewise consistent.
\end{corollary}

\begin{corollary}[Sequential model preserves extrapolation]
\label{cor:frengressionseq_extrapolation}
If at each time point the static frengression model allows extrapolation on continuous treatments (\Cref{prop:extrap}), then the joint sequential model also admits extrapolation on continuous treatments at every time point.
\end{corollary}
\begin{proof}[Proof of \Crefrange{cor:frengressionseq_correctness}{cor:frengressionseq_extrapolation}]
The proof proceeds by the induction on $t$.
For the base case $t=0$, the claims reduce to the static frengression result.
Now suppose the statement is true for all time steps up to $t-1$. Applying the static result to $\varphi_t$ establishes the properties for time $t$.
%
\end{proof}

\subsection{Sampling}
The sampling works similarly to the static frengression with sequential chaining. Similar to \Cref{sec:sampling_static}, we provide sampling from $\varphi_s^\ast$ for the following scenarios.
\begin{enumerate}
    \item Simulate $\widehat{C}$: Draw $\varsigma\sim \cN(0,I_{d_c})$ and apply the transformation $g_c^\ast(\varsigma)$.
    
    \item Simulate $(\widehat{Z},\widehat{X})\in \R^{T\times(d_z+d_x)}$: Specify $c$ (can be simulated from $g_c^\ast(\varsigma)$) and sample $\epsilon_0\sim\cN(0,I_{d_z+d_x})$; apply $g_0^\ast(\cdot)$ on $(c,\epsilon_0)$ to get $(\widehat{Z}_0, \widehat{X}_0)$. Then, for $t=1,\ldots, T-1$: simulate $\epsilon_t\sim\cN(0,I_{d_z+d_x})$, then $(\widehat{Z}_t, \widehat{X}_t) = g_t^\ast(c,\widehat{\overline{Z}}_{t-1},\widehat{\overline{X}}_{t-1},\epsilon_t)$. Putting them together we obtain $(\widehat{Z},\widehat{X})\in \R^{T\times(d_z+d_x)}$.
    
    \item Simulate $\widehat{Y}(\overline{x}_{T-1})$: specify $c$ (which can be simulated from $g_c^\ast$). For $t=0,\ldots, T-1$, sample $\eta_t \sim \cN(0,I_{d_y})$, then apply the transformation $f_t^\ast$ with sampled $\eta_t$, specified $c$, and specified $\overline{x}_t$ to get $\widehat{Y}_t=f_t^\ast(c,\overline{x}_t,\eta_t)$. Putting them together, we obtain $\widehat{Y} = (\widehat Y_0,\ldots,\widehat Y_{T-1})$.

    \item  Simulate from $\varphi_s$ to obtain $(\widehat{C},\widehat{\overline{Z}}_t,\widehat{\overline{X}}_t,\widehat{\overline{Y}}_t) \in \R^{d_c + T\times(d_z+d_x+d_y)}$: follow the instructions above to obtain $\widehat C$, $\widehat{Z}$, $\widehat{X}$. For $t=0,\ldots, T-1$, sample $\xi_t\sim\cN(0,I_{d_y})$, and apply the transformation $h_t^\ast $ on sampled  ($\widehat C,\widehat{\overline{Z}}_{t}, \widehat{\overline{X}}_t ,\xi_t$) to get $\widetilde\eta_t =h_t^\ast(\widehat C, \widehat{\overline{Z}}_{t}, \widehat{\overline{X}}_t, \xi_t)$. We obtain $\widehat Y_t$ via $\widehat Y_t=f_t^\ast(\widehat C,\widehat{\overline{X}}_t,\widetilde\eta_t)$.
\end{enumerate}

\begin{proposition}[Simulate data from joint distribution]
\label{prop:longitudinal_joint}
We take sequential conditional exchangeability and that \Cref{ass:correct_model} holds. The simulation procedure stated above generates a draw $\P_\varphi$ of $(C,Z,X,Y)$ whose distribution equals the observed joint distribution of $(C,Z,X,Y)$.
\end{proposition}
\begin{proof}
We proceed by induction on the time index $t$. 
For $t=0$, first sample $\widehat C\sim\P_{g_c}$.
Conditional on $\widehat C$, draw
$(\widehat Z_0,\widehat X_0)\sim \P_{Z_0X_0\mid C=\widehat C}$ by applying $g_0^\ast(\cdot)$ on $(\widehat C, \epsilon_0)$, where $\epsilon_0\sim\cN(0,I_{d_z+d_x})$. Draw $\xi_0\sim \cN(0, I_{d_y})$, then apply $h_0^\ast$ on sampled $\widehat C, \widehat Z_0,\widehat X_0, \xi_0$ to get $\widetilde{\eta}_0$. Obtain
$\widehat Y_0$ by applying $f_0^\ast(\cdot)$ on $(\widehat C, \widehat{X}_0, \widetilde{\eta}_0)$ . Guaranteed by \Cref{prop:population_estimates}, $(\widehat{Z}_0, \widehat{X}_0) \sim \P_{Z_0X_0|C=\widehat{C}}$, $\widehat Y_0\sim \P_{Y_0\mid C=\widehat C,Z_0=\widehat Z_0,X_0=\widehat X_0}$ . This takes care of the base case.


Suppose that for $t\geq 1$, $\big(\widehat C,\widehat{\overline Z}_{t-1},\widehat{\overline X}_{t-1},
        \widehat{\overline Y}_{t-1}\big)
  \sim
  \big(C,\overline Z_{t-1},\overline X_{t-1},\overline Y_{t-1}\big).$
We now extend to time $t$ by drawing 
$(\widehat Z_t,\widehat X_t)$ from the distribution of $(Z_t,X_t\mid C=\widehat C,
  \overline Z_{t-1}=\widehat{\overline Z}_{t-1},
  \overline X_{t-1}=\widehat{\overline X}_{t-1})$
and then
$\widehat Y_t$ is sampled from $Y_t\mid C=\widehat C,
  \overline Z_t=\widehat{\overline Z}_t,
  \overline X_t=\widehat{\overline X}_t$,
  using the procedure stated in the above simulation settings.
Applying the chain rule yields
$\big(\widehat C,\widehat{\overline Z}_t,\widehat{\overline X}_t,
        \widehat{\overline Y}_t\big)
  \sim
  \big(C,\overline Z_t,\overline X_t,\overline Y_t\big),$
which completes the inductive step.
%
\end{proof}

\section{Numerical Examples on Supported Settings}
\label{sec:synthetic_exp}
In this section, we evaluate frengression on synthetic benchmarks under both static and time‑varying treatment regimes. Our evaluation focuses on two key aspects:
\begin{enumerate}
    \item \emph{Estimation and inference.} We evaluate how precisely frengression recovers the underlying data-generating process, especially the marginal interventional quantities. To establish ground truth, we either use existing datasets or simulate new ones with known interventional distributions, employing two packages built on frugal parameterisation: the \texttt{causl} package \citep{causl} for static treatment settings, and the \texttt{survivl} package \citep{survivl} for longitudinal and survival data. Simulating data this way gives us exact control over the marginal quantity targets.

    \item \emph{Simulation performance.} We assess frengression’s ability to sample new data that both preserves the original data’s properties and matches the specified marginal interventional distributions. By comparing the generated samples with the original targets, we verify that our model can capture and replicate complex distributional patterns.
\end{enumerate}

In the experiments that follow, we compare frengression against leading, specialised methods in each domain. While these competitors excel in their specific settings, none combine the estimation and simulation on static treatments, longitudinal dynamics, survival outcomes, and continuous treatments in a single framework. This breadth underscores frengression’s unique, comprehensive applicability. The code for reproducibility can be found at \href{https://github.com/xwshen51/frengression}{https://github.com/xwshen51/frengression}.

\subsection{Binary Intervention}

Estimating a treatment effect under a binary intervention via frengression does not require propensity score estimation or adjustment; instead, it directly models the marginal interventional quantity of interest. This contrasts with methods that average outcome‐regression predictions over the covariate distribution, whose performance can suffer when certain covariate combinations are rare or poorly represented. By focusing on the marginal causal parameter, frengression avoids the instability caused by extreme propensity score values. We assess its robustness under weak overlap scenarios, especially in finite samples, to demonstrate its ability to recover causal effects even when classical overlap conditions are violated. We first investigate how frengression handles weak overlap using synthetic data.

We simulate \(p = 10\) baseline covariates \(Z = (I, C)\), where the instrumental variables \(I \in \mathbb{R}^{5}\) and the confounders \(C \in \mathbb{R}^{5}\) are generated independently as standard normals. Treatment assignment follows $X\mid Z\sim\mathrm{Bernoulli}\big(\pi(Z)\big)$ with $\pi(Z)=\operatorname{expit}(\beta_I^\top I+\mathbf{1}^\top C)$, $\beta_I \in \R^5$ and the outcome is $\Yx\sim\mathcal{N}\big(2x,\,1\big)$, with Gaussian copulae between each confounder $C_j$, $j=1,\ldots,5$, and $Y$ with correlation $\rho = 2\operatorname{expit}(1)-1\approx 0.46$. By increasing the magnitude of the instrumental coefficients $\beta_I$, we push the propensity scores towards 0 or 1, creating a weak‐overlap scenario.

For each iteration, we simulate $N=5000$ observations and repeat the experiment for $K=30$ iterations.  We report bias, mean absolute error (MAE), and root mean squared error (RMSE) of the ATE estimates. We compare frengression against the augmented inverse probability weighted (AIPW) estimator (with logistic regression to estimate propensity score and Random Forests as the outcome regression model), Dragonnet, and CausalEGM. TARNet, CFRNet and CEVAE demonstrated sensitivity to hyperparameters and performed poorly, hence we exclude these results. Furthermore, these methods were primarily proposed for conditional treatment effect quantities, rather than marginal causal effects.

  
Results in \Cref{tab:synthetic_binary_results} show that the AIPW estimator’s variance increases as overlap worsens, while frengression remains stable. Even without per-dataset hyperparameter tuning, frengression matches or outperforms Dragonnet and CausalEGM in bias, MAE, and RMSE across all regimes. We present the experiment details in \Cref{sec:binary_exp_details}.

\begin{table}[ht]
\centering
\begin{tabular}{rrrrrr}
\toprule
 &  & Frengression & AIPW & CausalEGM & Dragonnet \\
$\beta_I$ & metric &  &  &  &  \\
\midrule
\multirow[t]{3}{*}{0.0} & RMSE & $\mathbf{0.056}$ & 0.290 & 0.151 & 0.070 \\
 & Bias & $\mathbf{-0.000}$ & 0.225 & $-0.117$ & 0.031 \\
 & MAE & $\mathbf{0.048}$ & 0.253 & 0.130 & 0.056 \\
\cline{1-6}
\multirow[t]{3}{*}{0.5} & RMSE & \textbf{0.076} & 0.290 & 0.149 & 0.114 \\
 & Bias & $\mathbf{-0.019}$ & 0.258 & $-0.124$ & 0.042 \\
 & MAE & $\mathbf{0.059}$ & 0.258 & 0.132 & 0.100 \\
\cline{1-6}
\multirow[t]{3}{*}{1.0} & RMSE & 0.071 & 0.849 & 0.193 & \textbf{0.055} \\
 & Bias & $\mathbf{-0.015}$ & 0.480 & $-0.173$ & \textbf{0.015} \\
 & MAE & 0.056 & 0.480 & 0.173 &\textbf{ 0.039} \\
\cline{1-6}
\multirow[t]{3}{*}{1.5} & RMSE & $\mathbf{0.086}$ & 0.690 & 0.262 & 0.155 \\
 & Bias & $\mathbf{-0.028}$ & 0.530 & $-0.245$ & 0.080 \\
 & MAE & $\mathbf{0.069}$ & 0.543 & 0.245 & 0.117 \\
\cline{1-6}
\multirow[t]{3}{*}{2.0} & RMSE & \textbf{0.084} & 0.524 & 0.287 & 0.113 \\
 & Bias & $\mathbf{-0.026}$ & 0.437 & $-0.274$ & 0.028 \\
 & MAE & \textbf{0.071} & 0.466 & 0.274 & 0.095 \\
\cline{1-6}
\bottomrule
\end{tabular}
\caption{Estimation performance across different degrees of overlap.}
\label{tab:synthetic_binary_results}
\end{table}

\subsection{Continuous Treatment}
In the context of continuous treatment, obtaining the average dose-response function (ADRF), defined as $\mu(x) = \E_Z\big[\E(\Yx\,|\, Z)\big]$, is typically of interest. To evaluate the models' performance, we use RMSE
and mean absolute percentage error (MAPE).

We follow \citet{liu2024encoding} in the synthetic experiments for continuous treatment, using the same data: two synthetic datasets, \citet{hirano04continuous} and \citet{sun2015causal}; and one semi-synthetic dataset, Twins, containing weights, mortality and other covariates of all 71,345 pairs of twins born in the
USA between 1989--1991. For the two synthetic datasets, we trained the models on $1000$ samples. The detailed information can be found in the supplementary appendix of \citet{liu2024encoding}. We present the key information in \Cref{tab:continuous_dataset}.

\begin{table}[!tb]
\centering
\resizebox{\textwidth}{!}{%
\begin{tabular}{c||c|c|c}
\toprule
               & \textbf{Twins} 
               & \textbf{Hirano \& Imbens} 
               & \textbf{Sun} \\
\midrule
$p$            & 50  
               & 200  
               & 200 \\
\hline
$Z$            & \makecell[c]{Real data derived from \\ all births in the USA \\between 1989--1991.  }
               & $(Z_1,\ldots, Z_{p}) \overset{\rm i.i.d.}{\sim} \operatorname{Exp}(1)$ 
               & $(Z_1,\ldots, Z_{p}) \overset{\rm i.i.d.}{\sim} \cN(0,1)$  \\
\hline
$\Yx\cmid Z$            & \makecell[c]{
                   $\Yx=-\dfrac{2}{1+e^{-3x}} + z^\top\gamma + \epsilon$,\\
                   $\gamma_i\sim\mathcal{N}(0,0.025^2),\ \epsilon\sim\mathcal{N}(0,0.25^2)$
                }
               & $ \mathcal{N}\big(x+(Z_0+Z_2)e^{-x(Z_0+Z_2)},\,1\big)$
               & $\mathcal{N}\big(x+Z_1+\cos(Z_2)+Z_5^2+Z_6-\tfrac12,\,1\big)$ \\
\hline
$X\cmid Z$            & weights (observed)
               & $\exp(Z_0+Z_1)$
               & $\mathcal{N}\big(-2\sin(2Z_1)+Z_2^2+Z_3+\cos(Z_4)-\tfrac56,\,1\big)$ \\
\hline
$\mu(x)$        & $\displaystyle -\frac{2}{1+e^{-3x}} + \E\big[Z^\top\gamma\big]$
               & $\displaystyle x + \frac{2}{(1+x)^3}$
               & $\displaystyle x + 0.5 + e^{-0.5}$ \\
\bottomrule
\end{tabular}
}
\caption{Key information of the three datasets in continuous treatments, including the dimension of covariates $p$, the data generation process of $Z$, $X$, $Y$ and the ground truth ADRF $\mu(x)$.}
\label{tab:continuous_dataset}
\end{table}

\begin{figure}[!tb]
    \centering
    \includegraphics[width=1\linewidth]{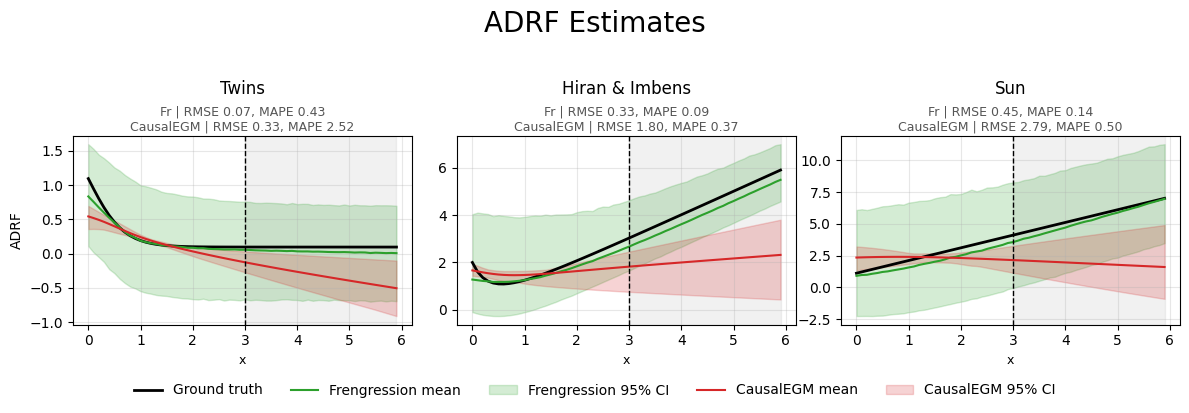}
    \caption{ADRF estimation, with RMSE and MAPE, frengression (Fr) and CausalEGM, across 30 simulations. We sample 1000 datapoints for training for each dataset. The grey area is not seen during training. We train both models for 1000 epochs in all simulations.}
    \label{fig:adrf_estimate}
\end{figure}

\Cref{fig:adrf_estimate} compares ADRF estimates from frengression and CausalEGM on the three benchmark datasets over $K=30$ simulations. We trained both models for $1000$ epochs in all experiments. For CausalEGM specifically, we used the best model configuration provided by their package. 

For each simulation $k$ and treatment level $x$, we compute
the average of 1000 draws from the learned frengression model as the estimated treatment effect on level $x$.  We then aggregate across simulations to obtain the ensemble estimate $\hat{\mu}(x) = \frac{1}{K}\sum_{k=1}^K \hat{\mu}^k(x)$, plotted as a solid green line for frengression and a red line for CausalEGM. 

To illustrate uncertainty, we add shaded bands corresponding to the 2.5\%--97.5\% interval. 
  For the $k$th simulation of frengression, let $\hat{q}^k_{\alpha}(x)$ be the empirical $\alpha$ quantiles of the 1000 samples from $\hat \P^{*,k}_{\Yx}$.  The band at $x$ is then
$$
  \bigg[\frac{1}{K}\sum_{k=1}^K \hat{q}^k_{2.5\%}(x),
        \frac{1}{K}\sum_{k=1}^K \hat{q}^k_{97.5\%}(x)\bigg].
$$
%
  Since CausalEGM produces point estimates only, we take the 2.5\% and 97.5\% sample quantiles of $\{\hat\mu^k(x)\}_{k=1}^K$ to form its shaded band.

We assessed each model’s ability to recover the dose-response function on a treatment range unseen during training by withholding treatments in the interval $(3,6]$ from the fitting process. After training, we estimated the ADRF over the full treatment range, including the held-out region $x\in (3,6]$, and compared the estimated ADRF with the ground truth. As shown in \Cref{fig:adrf_estimate}, frengression consistently outperformed CausalEGM across all three benchmark datasets, showing superior accuracy both within the training range and in the extrapolation interval, and did so without any additional hyperparameter tuning. The estimated 2.5\% and 97.5\% quantile bands accurately bound the true ADRF across all treatment levels. 

\subsection{Longitudinal Setting \& Survival Data}
\label{sec:longitudinal}
We demonstrate the modelling on longitudinal data. To the best of our knowledge, only a few deep learning frameworks have been proposed regarding inference and simulation on marginal causal effects. The recent work 
of \citet{shirakawa2024longitudinal}, which combines the Transformer architecture and LTMLE, handles only  estimation of causal effects on binary treatment in the longitudinal context.

We designed four synthetic settings to showcase the flexibility of frengression. Settings 1--3 are designed for survival analysis, and setting 4 is longitudinal. In settings 1--3, we denote $\widetilde Y_t $ as the residual lifetime after surviving to time $t$. If $\widetilde Y_t<1$, then $Y_t=1$, otherwise $Y_t=0$. The details of each setting are given below in \Cref{sec:sequential_exp_details}. We provide the summary of key information as below:
    
\begin{itemize} 
    \item Setting 1: $T=10$. $C$ is drawn from a Bernoulli distribution. We set the treatment $X_t$ to be binary and dependent on $(Z_t, C)$; the time-varying covariate $Z_t$ is dependent on previous treatment $X_{t-1}$ and $C$; the outcome $\widetilde{Y}(\overline{x}_t)\cmid C$ follows an exponential distribution.  
    \item Setting 2: $T=10$. $C$ is drawn from an exponential distribution with rate 1 
    and $Z_0$ from a standard normal; $X_0 \mid Z_0 \sim \operatorname{Bernoulli}\big(\operatorname{expit}(-0.5+ 0.5Z_0)\big)$; the treatment $X_t$ is still binary, and is dependent on both $Z_t$  and the previous treatment $X_{t-1}$; the time-varying covariate $Z_t$ depends on $X_{t-1}$, $C$ and $Z_{t-1}$. The Spearman correlation between $Y_t$ and $Z_{t-1}$ is $0.2$, and that for $Y_t$ and $Z_{t}$ given $Z_{t-1}$ is $0.3$. Compared to setting 1, setting 2 has more complicated dependency.
    \item Setting 3: $T=10$. $C$ is drawn from a Bernoulli distribution.  We set $X_t$ to be continuous from a normal distribution, dependent on $Z_{t-1}$ and $C$. We use  a Gaussian copula with correlation $0.4$ between $Z_t$ and $Y_t$.  In this setting the treatment is continuous.
    \item Setting 4: $T=5$. $C$ is drawn from a standard normal distribution. We set $X_t$ to be continuous from normal distribution, and the outcome $Y_t (\overline{x}_t)\,|\, C \sim \mathcal{N}(2x_t + x_{t-1}+0.5x_{t-2}+ C,\,1)$. This models the longitudinal scenario: at time step $t$ we are able to observe a continuous outcome rather than just the indicator of event occurrence.
\end{itemize}

\begin{figure}[!tb]
    \centering
    \includegraphics[width=1\linewidth]{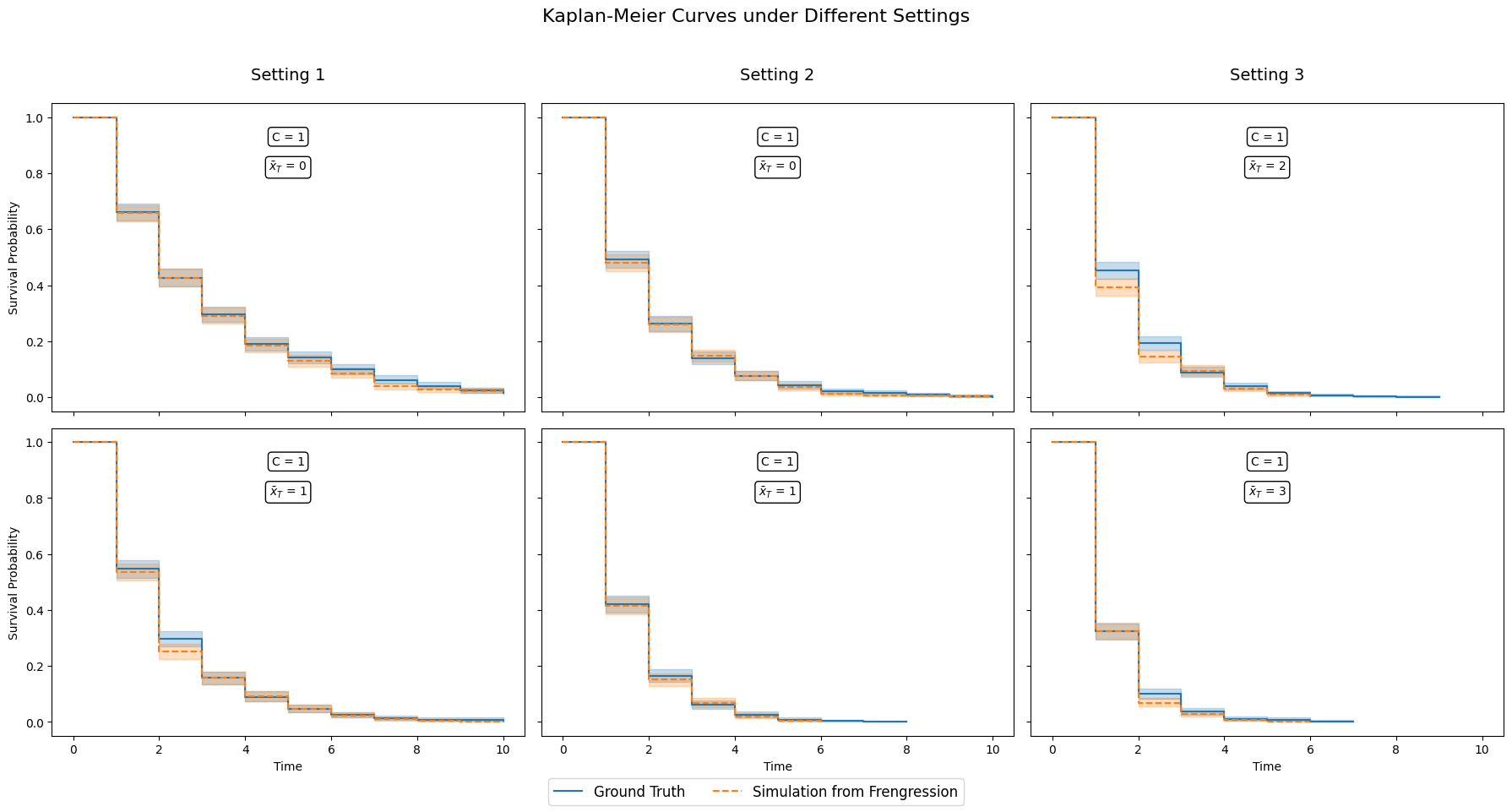}
    \caption{Kaplan-Meier curves of data generated from true data generating process and simulated from fitted frengression.}
    \label{fig:km_synthetic}
\end{figure}

In each setting, we simulate $n=6000$ samples for training. We then sample $1000$ datapoints from each fitted model. In \Cref{fig:km_synthetic}, we compare Kaplan-Meier curves fitted on the true data‐generating process with those simulated by frengression with intervention $x_t$ fixed the same across time. Note that frengression is trained on time‐varying covariates and interventions as described above. This supports our claim on the flexibility and comprehensiveness of the frengression model.

To create a performance baseline, we also fit Deep LTMLE using the \texttt{dltmle} package  \citep{dltmle} to setting 1.
Our causal estimand is the risk difference between always-treated and always-placebo regimens over the entire follow-up horizon
$
\theta = \mathbb{E} \,Y_{T-1}(\textbf{1}) - \mathbb{E} \,Y_{T-1}(\textbf{0}),
$
where, as before, $Y_t=1$ if the event has occurred by time $t$ and $0$ otherwise. If we set $T=5$, the true value of $\theta$ is $0.11$. Over 10 simulation runs, Deep LTMLE’s mean estimate was 0.13 (SD 0.03), and frengression’s was 0.09 (SD 0.04). These results are similar, and frengression achieves them without a separate targeting step or an explicit influence-function calculation, thus it can be applied to a broader class of causal estimands. We would also like to highlight the computational efficiency of frengression compared to \texttt{dltmle}. Hyperparameter tuning with \texttt{dltmle}'s default \texttt{Optuna} routine took about 30 minutes for 10 trials, plus 4 minutes to run 100 training epochs, whereas frengression completed 2000 epochs in 1.4 minutes on the same hardware.
\begin{figure}[!tb]
    \centering
    \includegraphics[width=.6\linewidth]{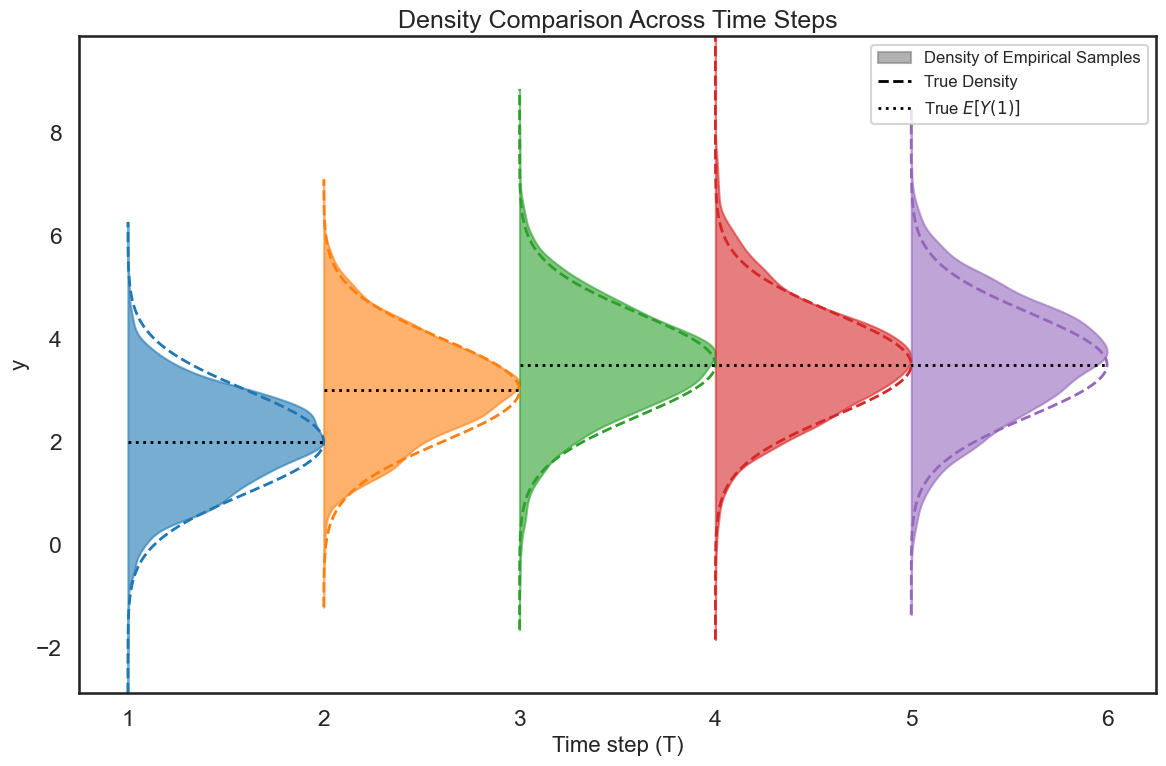}
    \caption{Density comparison of data simulated from frengression vs. true density, Setting 4.}
    \label{fig:synthetic_long}
\end{figure}

Figure \ref{fig:synthetic_long} shows the results for Setting 4, where we plot the density of 5000 points simulated from  the interventional distribution $\hat{\P}_{\Yx|C}$ estimated by frengression and set $x=\mathbf{1}$, $C=0$. As expected, the frengression fit closely matches the true interventional distribution, although the performance worsens slightly at the last time step $T=5$ due to growing complexity.

\subsection{Distributional Regression}

In practical applications, it is often necessary to simulate data under a pre-specified interventional distribution while preserving the empirical dependencies observed in the original sample. As a distributional‐regression framework, frengression can estimate the target marginal interventional distribution of the observed data; it also allows one to generate synthetic datasets that retain the observed correlation structure between the covariates and treatments, but target an arbitrarily specified marginal causal distribution.
\begin{figure}[!tb]
    \centering
    \includegraphics[width=1\linewidth]{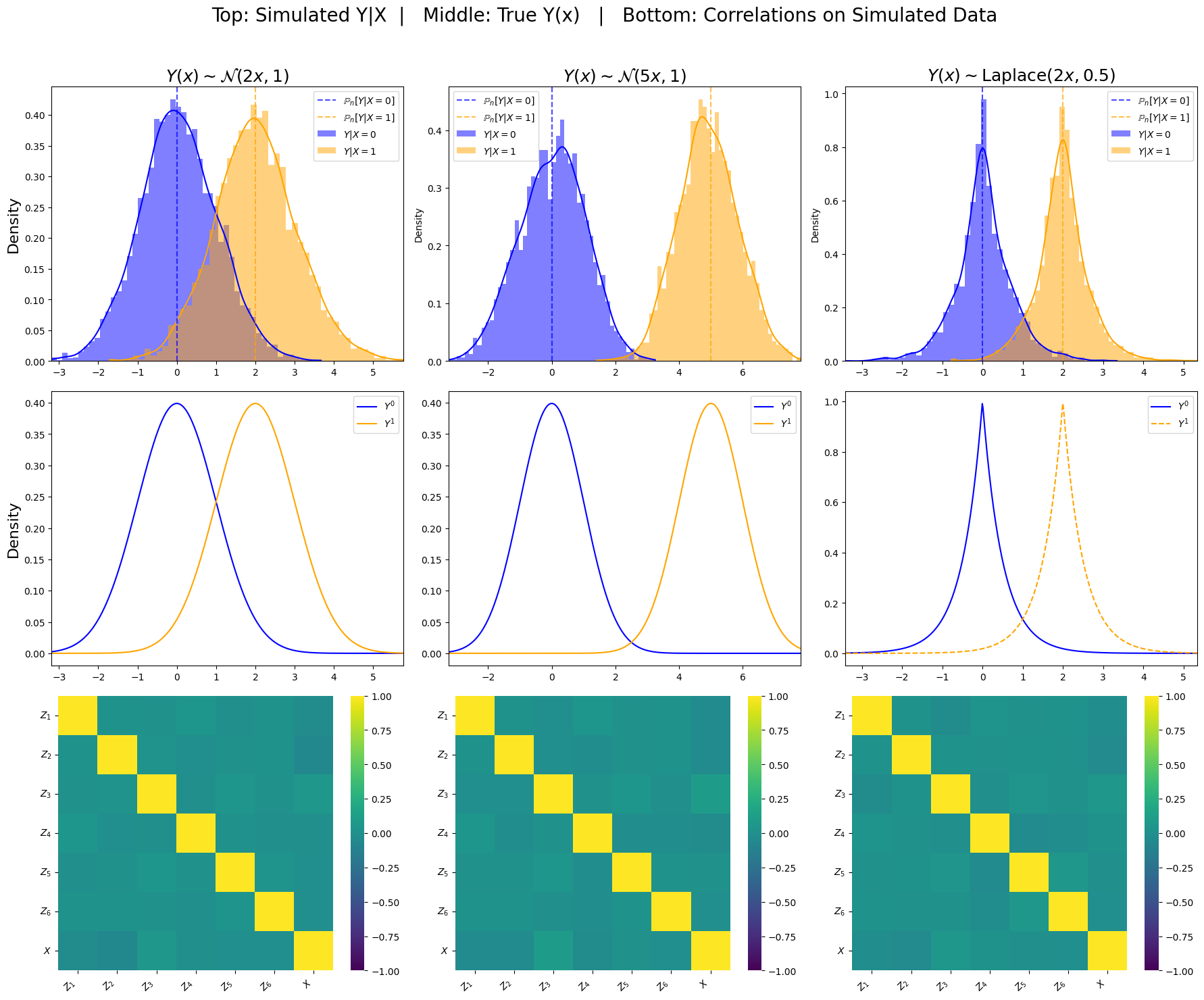}
    \caption{Marginal densities of simulations from frengression (top row), target (middle row) and correlation of covariates and treatment on simulations from frengression (bottom row). }
    \label{fig:distributional}
\end{figure}

In \Cref{fig:distributional}, we illustrate our method in the RCT (no confounding variables) context by comparing the empirical histogram from frengression’s simulations with the ground truth, and  showcase its ability to simulate data with user‐specified marginal causal distributions while preserving the joint dependency structure of the covariates and treatment.

First, we generate $n=5000$ samples from the true data‐generating process: the six covariates $Z\in\R^6$ are independently drawn from $\mathcal{N}(0,1)$, the treatment $X\sim \operatorname{Bernoulli}(0.5)$ as commonly seen in an RCT, and the intervened outcome $\Yx \sim N(2x,1)$, with a Gaussian copula between $(Z_5,Z_6)$ and $Y$ with Pearson correlation $\rho = 2\operatorname{expit}(1)-1\approx 0.46$. Fitting frengression to these samples recovers the true interventional distribution and the correlation between $X$ and $Z$ (left column).

Next, we replace the fitted $\Yx$ marginal with two new targets: $\mathcal{N}(5x,1)$ and $\operatorname{Laplace}(2x,0.5)$, and draw $n=5000$ simulated triples $(\widehat Z,\widehat X,\widehat Y)$ from the modified model (middle and right columns).  In \Cref{fig:distributional}, the top row shows the estimated marginal densities of $Y\mid X=0$ and $Y\mid X=1$ from the simulated data. Because there is no confounding, these curves should match the true interventional marginals.  The middle row displays the true target densities (the Gaussians or Laplace) for each intervention; we see that the simulation from the frengression fits align with these densities. The bottom row confirms that the original joint distribution of $(Z,X)$ remains unchanged under the modified interventional model. The results validate that in all cases, the simulated marginals match the newly specified objectives while preserving the baseline dependencies between $X$ and $Z$.

\section{Application on LEADER}
\label{sec:leader}
We illustrate frengression on the LEADER trial \citep{leader} to showcase its ability to simulate realistic, time-to-event data with complex covariate patterns. LEADER is a randomised, double‑blind, placebo‑controlled cardiovascular outcomes trial that enrolled 9340 patients with type 2 diabetes at high cardiovascular risk, randomising patients into liraglutide (up to 1.8 mg daily) or placebo, both on top of standard care, resulting in similarly sized treatment arms.  Patients were recruited from August 2010 through April and followed for $3.5$--$5$ years, with a median follow-up of 3.8 years. The primary endpoint is time to first major adverse cardiovascular event (MACE), including cardiovascular death, non-fatal myocardial infarction and non-fatal stroke. The original analysis found that MACE occurred in significantly fewer patients in the liraglutide group than in the placebo group \citep{marso2016liraglutide}.

To build our generative model, we select baseline covariates known to predict MACE. The baseline covariates include four binary indicators: gender, smoker status, carotid stenosis $>50\%$  on angiography, and diabetic nephropathy at screening. We have seven continuous variables at baseline: age, duration of diabetes (months), high-density lipoprotein (HDL) cholesterol, low-density lipoprotein (LDL) cholesterol, total cholesterol,  triglycerides, and serum creatinine. We also incorporated three key time‑varying biomarkers to capture longitudinal risk dynamics: HbA1c, body‑mass index (BMI), and estimated glomerular filtration rate (eGFR). The description of the variables can be found in \Cref{tab:leader_description}.

We partition the trial timeline into 6‑month intervals over a maximum of 60 months, yielding 11 discrete timepoints (i.e.~$t=0,1,2,\ldots 10$). We define each subject's event time as $k\in (0,\infty)$, where $k\in (0,10]$ indicates the  time MACE occurred in units of 6-month intervals, and $k>10$ indicates no event within the trial. With this notation, we define the event indicators $Y_t$, $t \in \{0,1,\ldots 10\}$ as 
$$
  Y_{t} =
  \begin{cases}
    1, & k \leq t,\\
    0, & \text{otherwise.}
  \end{cases}
$$
Here, we have $Y_t=1$ if and only if the event occurred by time $t$, otherwise $Y_t=0$.

Note that in the time-varying setting, covariates were recorded at different intervals: HbA1c was measured every 6 months, whereas BMI and eGFR were collected only at months 6, 12, 18, 30, 48, and 60 from the start of the trial. To retain as much information as possible, we impute missing BMI and eGFR values at intermediate timepoints using the most recent available measurement (last observation carried forward).

We exclude units missing any baseline covariates, leaving 9161 subjects. We take the logarithm of triglycerides and serum creatinine. We then normalise all continuous covariates by subtracting their mean values and dividing them  by their standard deviations. After fitting frengression to the processed dataset, we simulate 9161 synthetic records from the model for 5000 times.

\begin{figure}[!tb]
    \centering
    \includegraphics[width=1\linewidth]{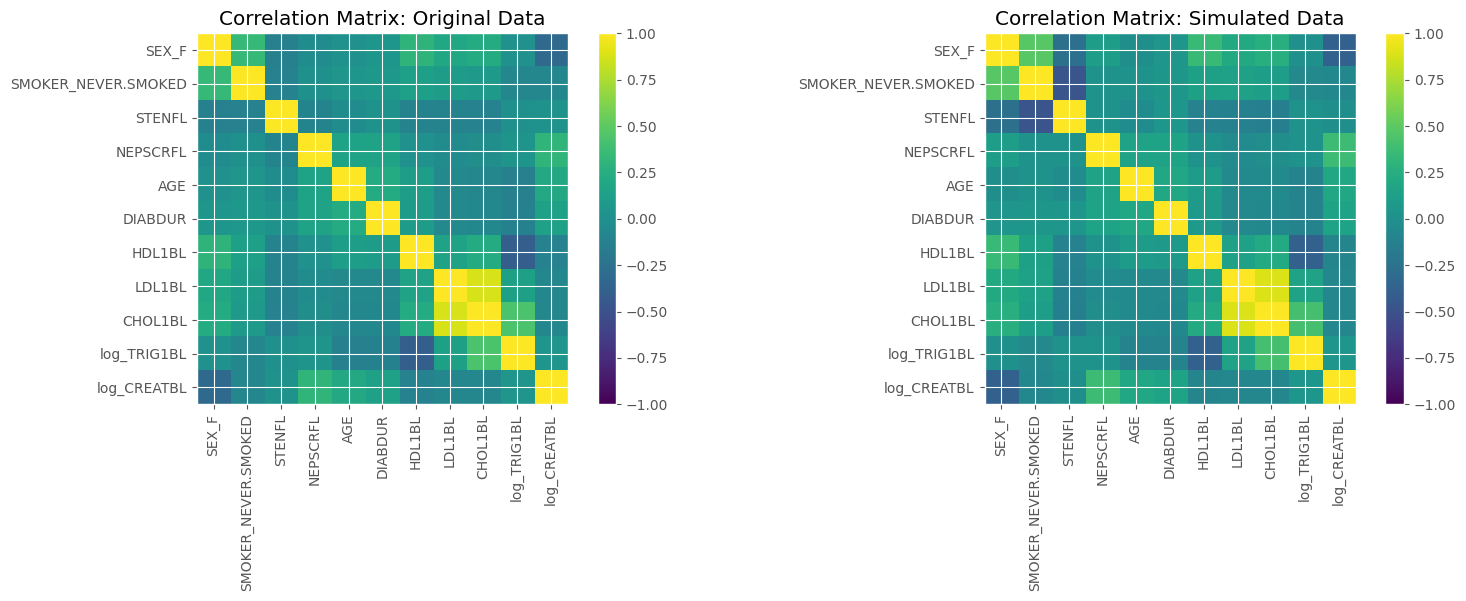}
    \caption{Correlation matrix of baseline covariates, real data vs. data simulated from fitted frengression. The heatmap illustrates the similarity of simulated data with real data regarding correlation coefficient.}
    \label{fig:corr_bsl}
\end{figure}

\begin{figure}[!tb]
    \centering
    \includegraphics[width=0.8\linewidth]{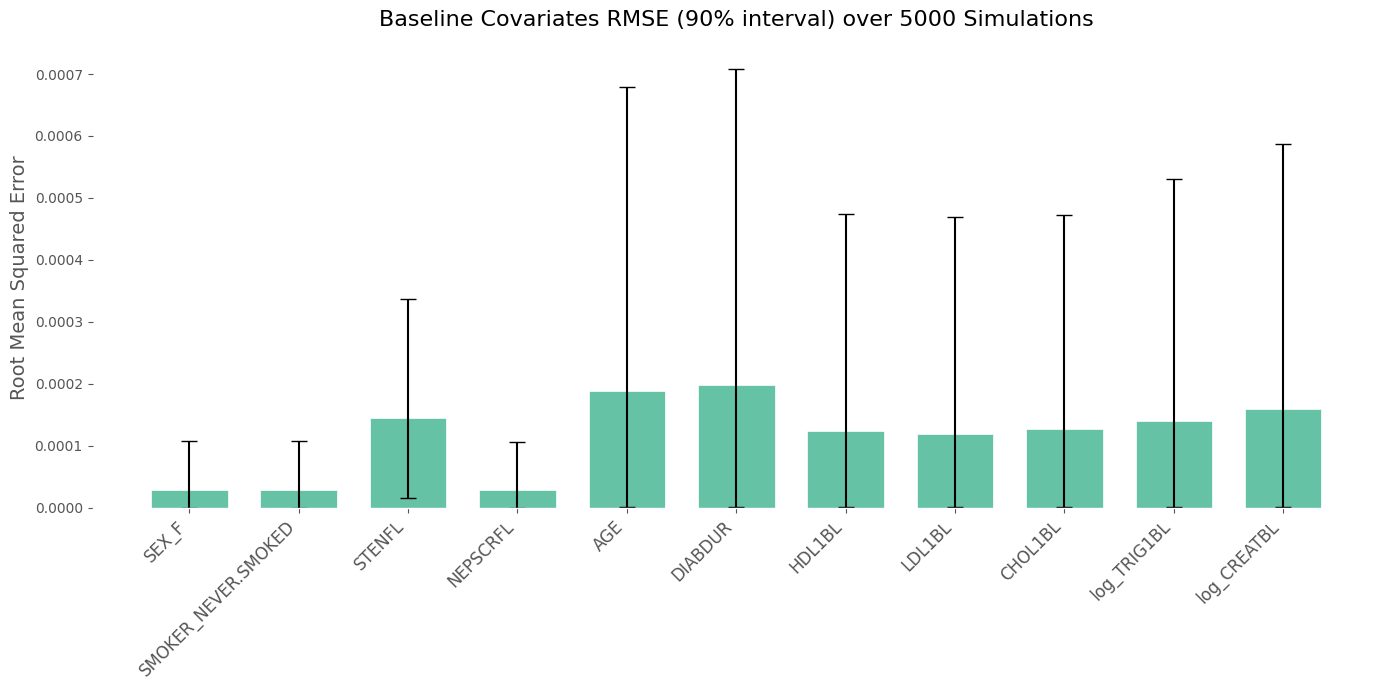}
    \caption{RMSE of each baseline covariate (continuous variables are normalised), with the maximum RMSE less than $0.001$ over 5000 simulations.}
    \label{fig:rmse_bsl}
\end{figure}

\begin{figure}[!tb]
    \centering
    \includegraphics[width=0.7\linewidth]{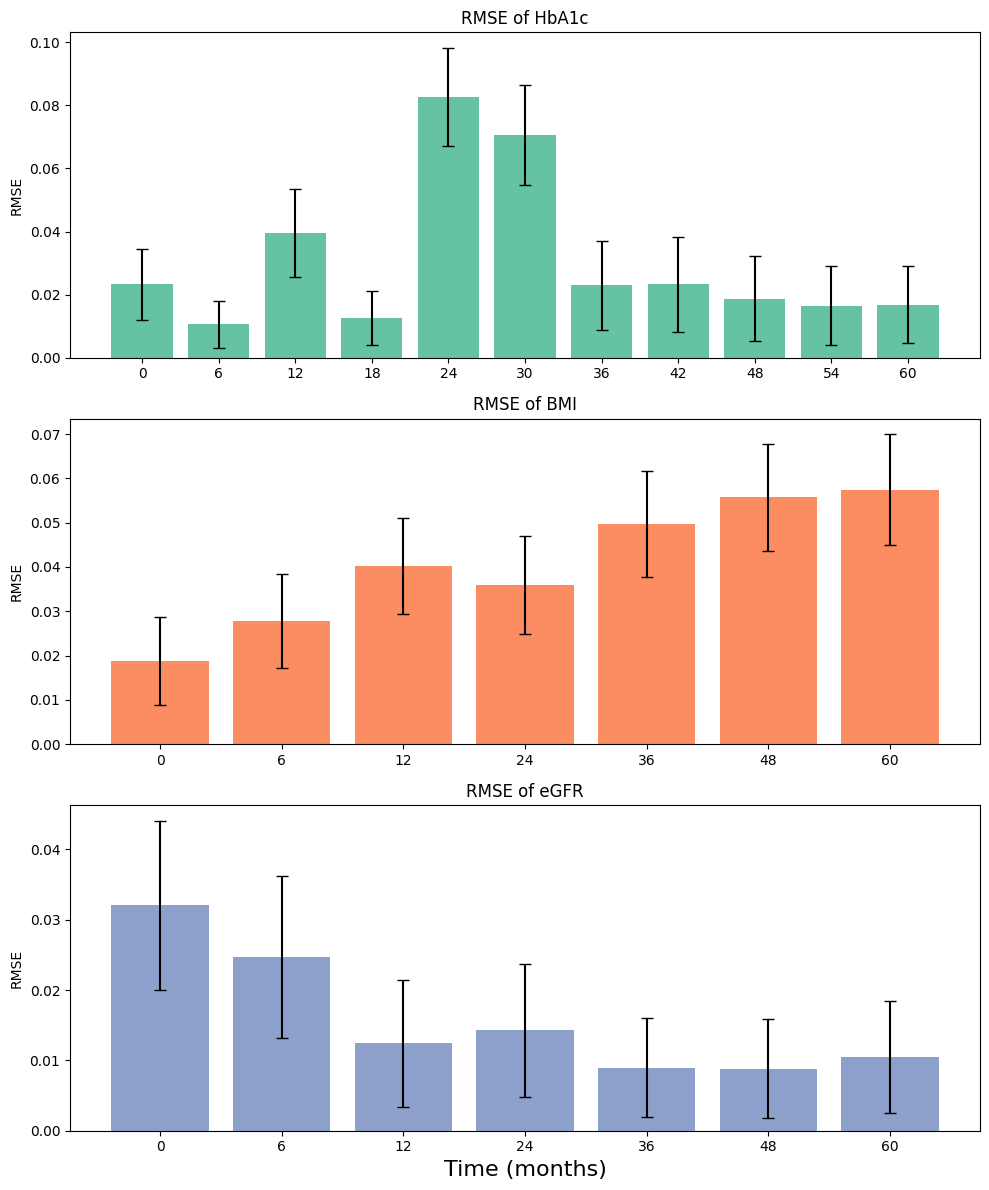}
    \caption{RMSE of each time varying covariate (normalised) across different time-points over 5000 simulations. All of them are of small scales (maximum $0.1$).}
    \label{fig:time-varying-rmse}
\end{figure}

\Crefrange{fig:corr_bsl}{fig:time-varying-rmse} visualise frengression’s empirical performance on data simulation. For each variable at each time point, we compute the RMSE between the mean of the simulated samples and the true mean. Overall, frengression accurately reproduces the linear correlations and achieves low RMSE across both baseline and time‑varying covariates. However, as shown in \Cref{fig:corr_bsl}, the simulated correlations of the STENFL variable (carotid $>50\%$ stenosis on angiography) with SEX\_F (gender) and SMOKER\_NEVER.SMOKED (indicating smoking status) deviate from that of the true data, reflecting the greater challenge of modelling binary variables using deep generative models compared to continuous ones.

To quantify how well the simulated data match the originals, we train logistic classifiers to distinguish real observations from simulated ones. On baseline covariates, the mean AUC is $0.50$ (range $0.47$--$0.52$), indicating classification that is nearly random. For the joint datasets (including baseline covariates, time‑varying covariates and MACE occurrence), the mean AUC is $0.55$ (range $0.52$--$0.57$), still reflecting substantial similarity between simulated and real distributions.

We also sample event occurrences from the fitted interventional distribution $\hat{\P}_{\Yx| C}$. Over 5000 simulations, the mean MACE occurrence rate in the placebo arm is $14.7\%$ with $90\%$ interval $[13.8\%, 15.6\%]$ (5th and 95th quantiles of the 5000 estimates), and in the liraglutide arm it is $13.0\%$ with $90\%$ interval $[12.2\%,13.9\%]$. These values match, to the reported precision, the $14.7\%$ in the placebo arm and $13.0\%$ in the liraglutide arm from real data, confirming that frengression faithfully captures both covariate structures and event frequencies.

We applied DLTMLE to the LEADER data for comparison; implementation details and hyperparameters are given in \Cref{sec:LEADER_experiment_details}. DLTMLE yields an estimated event rate of $15.4\%$ ($95\%$ CI: $[14.0\%, 16.5\%]$) under placebo and $12.9\%$ ($95\%$ CI: $[11.9\%, 14.0\%]$) under treatment. These deviate more from the observed trial rates than the corresponding frengression estimates. As DLTMLE is designed for estimation rather than generation, it cannot be evaluated in terms of simulation fidelity (e.g., RMSE or correlation structure in Figures \ref{fig:corr_bsl}--\ref{fig:time-varying-rmse}).

Estimation of $\E[Y(a)]$ and faithful data generation are complementary objectives. DLTMLE addresses the former. Frengression, in contrast, provides a unified approach: it produces consistent estimation while simultaneously learning a generative model capable of producing synthetic data.

\section{Discussion}
\label{sec:discussion}

The rapid growth of AI makes reliable causal inference more important than ever. Contributing to this, we introduce a deep generative approach that generates synthetic data with a specified marginal interventional distribution, while maintaining a realistic level of complexity. Such a data generation tool is helpful in different applications, including pre-experiment study design for randomised controlled trials, selecting an inference procedure, and stress testing AI methods \citep{giuffre2023harnessing,pezoulas2024synthetic}.

There are several promising extensions. As noted in \Cref{remark:privacy}, the variation independent decomposition in frengression can be directly adapted for differentially private simulation. Although differential privacy remains the gold standard for protecting individual records, its strict requirements often come at the expense of simulation accuracy in practice.  This means that suitable measures require a better understanding of the needs of end users \citep{yoon2020anonymization, kenny2021use,mccartan2023making}.

Methodologically, although frengression can generate data of subject-level trajectories under time-varying treatment, or any quantity conditional on a subset of covariates, its robustness and efficiency for individual treatment effects, such as those needed for dynamic treatment regimes, remain open questions. Even so, we find that centring on the target causal distribution explicitly brings clarity and consistency in both estimation and simulation.

At present, frengression’s extrapolation works only with continuous treatments; how best to extrapolate when treatments are discrete is to be explored. In addition, we provided consistency results for the generative model. Deriving convergence rates would help us build valid confidence intervals for the estimates that frengression produces.

\section*{Acknowledgements}
Linying Yang is supported by the EPSRC Centre for Doctoral Training in Modern Statistics and Statistical Machine Learning (EP/S023151/1) and Novartis.  Robin J.~Evans is supported by GSK.  We acknowledge funding administered by the Danish National Research Foundation in support of the Pioneer Centre for SMARTbiomed. 


\eject

\appendix

\section{Auxiliary results}

This section collects auxiliary results that support the main theorems.

\subsection{Distribution consistency from energy distance convergence}

\begin{theorem}[Weak convergence from convergent energy distance]
\label{thm:convergence_by_convergent_ed}
Let $\Q$ be a probability measure on $\R^d$, and let $\P^n$ be a random
probability measure on $\R^d$. Suppose that $\P^n$ and $\Q$ satisfy
\Cref{ass:tail}, and that
\[
    \ed(\P^n,\Q)\to 0
    \qquad\text{in probability}.
\]
Then, for every bounded continuous function $f$,
\[
    \int f\,d\P^n \to \int f\,d\Q
    \qquad\text{in probability}.
\]
Equivalently, $\P^n \Rightarrow \Q$ in probability.
\end{theorem}

\begin{proof}
Fix a bounded continuous function $f$. By the subsequence principle for
convergence in probability, it suffices to show that every subsequence of
$\left\{\int f\,d\P^n\right\}$ contains a further subsequence that converges
almost surely to $\int f\,d\Q$.

Let $\{\P^{n_k}\}$ be an arbitrary subsequence. Since
\[
    \ed(\P^n,\Q)\to 0
    \qquad\text{in probability},
\]
there exists a further subsequence, not relabelled, such that
\[
    \ed(\P^{n_k},\Q)\to 0
    \qquad\text{almost surely}.
\]
Fix an outcome in the corresponding probability-one set. Then
$\{\P^{n_k}\}$ is a deterministic sequence of probability measures satisfying
$\ed(\P^{n_k},\Q)\to 0$.

By \Cref{ass:tail}, for any $\varepsilon>0$ there exists $R>0$ such that
$\P^{n_k}(\|X\|>R)<\varepsilon$ uniformly in $k$, hence $\{\P^{n_k}\}$ is tight. Hence, by Prokhorov's
theorem \citep[Theorem~5.1]{billingsley2013convergence}, every subsequence of
$\{\P^{n_k}\}$ contains a further subsequence that converges weakly to some
probability measure. Let $\{\P^{n_{k_j}}\}$ be such a subsequence, and write
\[
    \P^{n_{k_j}} \Rightarrow \pi
\]
for some probability measure $\pi$ on $\R^d$.

Let $\varphi_P(t)=\int e^{i\langle t,x\rangle}\,dP(x)$ denote the
characteristic function of a probability measure $P$. Since
$\P^{n_{k_j}} \Rightarrow \pi$, we have
\[
    \varphi_{\P^{n_{k_j}}}(t)\to \varphi_\pi(t)
    \qquad\text{for every } t\in\R^d.
\]
Moreover, \Cref{ass:tail} implies finite first moments, so the
characteristic-function representation of the energy distance applies
\citep[Proposition~1]{szekely2013energy}:
\[
    \ed(\P^{n_{k_j}},\Q)
    =
    \frac{1}{c_d}
    \int_{\R^d}
    \frac{\left|\varphi_{\P^{n_{k_j}}}(t)-\varphi_\Q(t)\right|^2}{\|t\|^{d+1}}\,dt.
\]
Since $\ed(\P^{n_{k_j}},\Q)\to 0$, the right-hand side converges to zero.
Therefore, after passing to a further subsequence if necessary,
\[
    \varphi_{\P^{n_{k_j}}}(t)\to \varphi_\Q(t)
    \qquad\text{for almost every } t\in\R^d.
\]
Combining this with the pointwise convergence to $\varphi_\pi$ yields
\[
    \varphi_\pi(t)=\varphi_\Q(t)
    \qquad\text{for almost every } t\in\R^d.
\]
Since characteristic functions are continuous, this implies
$\varphi_\pi=\varphi_\Q$ everywhere, and hence $\pi=\Q$.

Thus every subsequence of $\{\P^{n_k}\}$ has a further subsequence converging
weakly to $\Q$. Hence the deterministic sequence $\{\P^{n_k}\}$ converges
weakly to $\Q$. In particular,
\[
    \int f\,d\P^{n_k}\to \int f\,d\Q
    \qquad\text{almost surely}.
\]
Since the original subsequence was arbitrary, the subsequence principle for
real-valued random variables implies
\[
    \int f\,d\P^n \to \int f\,d\Q
    \qquad\text{in probability}.
\]
Because $f$ was arbitrary, this proves $\P^n \Rightarrow \Q$ in probability.
\end{proof}

\section{Main proofs}
\subsection{Proof of \Cref{lem:identifiability}}
\label{sec:proof_identifiability}

\begin{proof}
When the conditions in \Cref{ass:identification} hold, we have 
\begin{align}
p_{\Yx}(y) & = p_{\Yx|X_0}(y\cmid x_0) \label{eqn:identify_1}\\
&= p_{Y\hspace{-.75pt}(X_0,x_1)|X_0}(y\cmid x_0)\label{eqn:identify_2}\\
&=\int_{\mathcal{Z}} p_{Y\hspace{-.75pt}(\Xpaz,x_1)|\XpZ}(y\cmid x_0,z) \, p_{\ZIXp}(z\cmid x_0)\,{\rm d} z \label{eqn:identify_3}\\
&=\int_{\mathcal{Z}} p_{Y\hspace{-.75pt}(\Xpaz,x_1)|\XpZXc}(y\cmid x_0,z,x_1) \, p_{\ZIXp}(z\cmid x_0)\,{\rm d} z \label{eqn:identify_4}\\
&=\int_{\mathcal{Z}} p_{Y|\XpZXc}(y\cmid x_0,z,x_1)\, p_{\ZIXp}(z\cmid x_0)\,{\rm d} z, \label{eqn:identify_5}
\end{align}
where \eqref{eqn:identify_1} and \eqref{eqn:identify_4} are from the unconfoundedness assumption, while \eqref{eqn:identify_2} and \eqref{eqn:identify_5} are based on the consistency assumption. Equation \eqref{eqn:identify_3} is from the law of total probability; \eqref{eqn:identify_1}, \eqref{eqn:identify_3} and \eqref{eqn:identify_4} require the positivity assumption.
\end{proof}

\subsection{Proof of \Cref{prop:population_estimates}}
\label{sec:proof_population_estimates}
\begin{proof}
Since the energy score is a strictly proper scoring rule~\citep{gneiting2007strictly}, the objective function in \eqref{eqn:population_g} is minimised if and only if $g(\epsilon)\sim \P_\ZX$. By Assumption~\ref{ass:correct_model}, such a model exists, we hence have $g^\ast(\epsilon)\sim \P_\ZX$. According to Proposition 1 of \citet{shen2024engression} and \Cref{ass:correct_model}, we have $f^\ast(x,h^\ast(z,x,\xi))\sim\P_{Y|Z=z,X=x}$ for $\P_\ZX$-almost every $(z,x)$. Because $\bar\eta$ and $\bar\eta'$ are formed from the same $\bar X$, the second term of \eqref{eqn:population_fh} is an integrated conditional energy score over $\bar X\sim\P_X$, so by strict propriety it is minimised if
\begin{equation}\label{eq:pf1}
    h^\ast(Z_x,x,\xi)\sim \cN(0,I_{d_y})\text{ when }Z_x\sim \P_{Z\mid X_0=x_0},
\end{equation}
for $\P_X$-almost every $x=(x_0,x_1)$. By \Cref{ass:correct_model}, a single $(f',h')$ satisfies both the conditional fit and this reference-noise condition, so the two terms of \eqref{eqn:population_fh} are minimised by a common solution.

By positivity (\Cref{ass:positivity}), the conditional fit established above for $\P_\ZX$-almost every $(z,x)$ holds, for $\P_X$-almost every $x$, at $\P_{Z\mid X_0=x_0}$-almost every $z$. The identification formula in \Cref{lem:identifiability},
\[
    p_{\Yx}(y) = \int_{\mathcal{Z}} p_{Y|\XpZXc}(y\cmid x_0,z,x_1) \, p_{\ZIXp}(z\cmid x_0) \, \mathrm{d} z.
    \]
suggests that to sample from $\P_{\Yx}$ for $\P_X$-almost every $x=(x_0,x_1)$, we can
\begin{enumerate}[label=(\roman*)]
    \item first sample $Z$ from $\P_{Z|X_0=x_0}$, and then 
    \item sample $Y$ from $\P_{Y|Z,X=x}$ by setting $\eta^\ast:=h^\ast(Z,x,\xi)$ and applying $f^\ast(x,\eta^\ast)$.
\end{enumerate}
According to \eqref{eq:pf1}, we know $\eta^\ast$ follows the standard Gaussian for $\P_X$-almost every $x$. Thus, we can simplify the above sampling procedure to directly sampling $\eta^\ast\sim\cN(0,I_{d_y})$ and applying $f^\ast(x,\eta^\ast)$. This yields the last desired result that
$f^\ast(x,\eta)\sim \P_{\Yx}$ for $\P_X$-almost every $x$.
\end{proof}

\subsection{Convergence results}
\label{sec:proof_consistency}

The energy distance $\ed(\P,\Q)=\|\mu_\P-\mu_\Q\|_{\mathcal H}^2$ uses the mean embedding $\mu_\P$ in the RKHS $\mathcal H$ of the energy kernel $k(u,v)=\tfrac12(\|u\|+\|v\|-\|u-v\|)$. On empirical inputs it has two estimators: the plug-in V-statistic $\ed_V\ge 0$, whose square root obeys the triangle inequality, and the unbiased U-statistic $\ed_U$, the training loss, which can be negative. We write $\ed$ throughout, meaning $\ed_V$ on empirical inputs, and $\ed_U$ only for the training loss and the concentration bounds (\Cref{lemma:finite_ed_second,lemma:concentration_inequality}); the two differ by lower-order self-terms, absorbed below.

We first note a consequence of \Cref{ass:tail}.

\begin{lemma}
[Finite second moment]
\label{lem:finite_moment}
Suppose \Cref{ass:tail} holds for random variable $X\sim \P_\theta$, $\forall\theta \in \Theta$ and $X\sim \Q$. Then, here exists a constant $M>0$ such that
$$\mathbb{E}_{\mathbb{P}_\theta}\|X\|^2 \le M^2
\quad\text{and}\quad
\mathbb{E}_{\mathbb{Q}}\|X\|^2\le M^2.
$$
\end{lemma}

\begin{proof}[Proof of \Cref{lem:finite_moment}]
    For $X\sim \P$ with $\P$ a probability measure satisfying \Cref{ass:tail}, we have
    \begin{align*}
        \E\|X\|^2 =\int_0^\infty  2u \Pr(\|X\|>u)\,\mathrm{d}u\le 2C\int_0^\infty u\operatorname{exp}(-\kappa u)\mathrm{d}u = \frac{2C}{\kappa^2}.
    \end{align*}
The claim follows with $M^2 =  \frac{2C}{\kappa^2}$.
\end{proof}

With \Cref{ass:tail} satisfied, we also have:
\begin{lemma}[Concentration inequality with energy distance]
\label{lemma:concentration_inequality}
    Let $\P$ be a probability measure on $X\subseteq \mathbb{R}^d$, $\P^n$ be the empirical measure obtained from $n$ independently and
identically distributed samples of $\P$, i.e.~$\P^n=\frac1n\sum_{i=1}^n\delta_{x_i}$, $x_i \overset{\rm i.i.d.}{\sim}\P$.

For every $\delta>0$, with probability at least $1-\delta-q_{R_n}$:
$$
\ed_U(\P,\P^n)< \frac{C_1M}{\sqrt{n(1-q_{R_n})}}+4{R_n}\sqrt{\frac{2}{n}\log\Big(\frac1\delta\Big)},
$$ for some constant $C_1>0$, where $R_n = \frac{3\log n}{\kappa}$, $q_{R_n}:=\Pr(\max_{1\leq i \leq n}\|x_i\|> R_n)$.
We also have:
$$
q_{R_n} \leq  \frac{C}{n^2}.
$$
\end{lemma}

To prove \Cref{lemma:concentration_inequality}, we first prove the following:
\begin{lemma}[Finite second moment of $\ed_U(\P,\P^n)$]
\label{lemma:finite_ed_second}
For a probability measure $\P$ that satisfies \Cref{ass:tail} and \Cref{lem:finite_moment}, where $\P^n$ is the empirical distribution, there exists a constant $C_0>0$, such that $\E[\ed_U(\P,\P^n)^2] \leq C_0M^2/n$.
\end{lemma}
\begin{proof}
We have
\begin{align*}
\ed_U(\P,\P^n) =\frac{2}{n}\sum_{i=1}^n\E_{X\sim \P}\|X-x_i\|- \E_{X,X'\sim \P}\|X-X'\|-\frac{1}{n(n-1)}\sum_{j=1}^n\sum_{\substack{k=1\\k\neq j}}^n\|x_j-x_k\|.
\end{align*}
We write $u(x) = \E_{X\sim \P}\|X-x\|$, $\nu = \E_{X,X'\sim \P}\|X-X'\|$, $U_n = \frac{1}{n(n-1)}\sum_{j=1}^n\sum_{\substack{k=1\\k\neq j}}^n\|x_j-x_k\|$. Clearly, $\E[U_n] =  \nu$.
We can thus write
\begin{align*}
    \ed_U(\P,\P^n) = \underbrace{\frac{2}{n}\sum_{i=1}^n \big(u(x_i)-\nu\big)}_{D_1}+\underbrace{\big(\nu-U_n\big).}_{D_2}
\end{align*}

For $D_1$:
$$
\Var(D_1) = \frac{4}{n}\Var\big(u(X)-\nu\big)= \frac{4}{n} \Var\big(u(X)\big).
$$
 By Jensen's inequality, we have $u(x)^2\leq \E_{X\sim \P}\|X-x\|^2$. Noting further that $\|X-x\|^2\leq2\|X\|^2+2\|x\|^2$,  and applying  \Cref{lem:finite_moment}, we get
\begin{align*}
    \E_{X\sim \P}[u(X)^2]\leq  4M^2,
\end{align*}
so that
$$
\Var(u(X))\leq \E[u(X)^2]\leq 4M^2
$$
and

$$
\Var(D_1) \leq \frac{4\cdot4M^2}{n}=\frac{16M^2}{n}.
$$

For $D_2$, 
using the standard variance formula for U-statistics of order two \citep[equation 5.18]{hoeffding1948ustatistics}, we obtain
$$
\Var(U_n) = \frac{4}{n}\sigma^2 + O\Big(\frac{1}{n^2}\Big),
$$ where $\sigma^2 = \Var\big[\E_{X'\sim \P} k(X,X')\big]
=\Var(u(X))\leq 4M^2$.
Thus,
$$
\Var(D_2) = \Var(U_n)\leq \frac{16M^2}{n} + O\Big(\frac{1}{n^2}\Big),
$$
which can be expressed as 
$$
\Var(D_2) \leq \frac{C_2M^2}{n},
$$ 
for some constant $C_2>0$. 
We can therefore obtain
$$
\E[\ed_U(\P,\P^n)^2] \leq 2\Var(D_1) + 2\Var(D_2),
$$
since $\E D_1 = \E D_2 = 0$.

Thus we obtain
\begin{align*}
\E[\ed_U(\P,\P^n)^2]\leq \frac{C_0M^2}{n}
\end{align*}
for some constant $C_0>0$.
\end{proof}

\begin{proof}[Proof of \Cref{lemma:concentration_inequality}]
\label{proof:concentration_inequality}
Let $\mathcal{A}_{R_n}$ denote the event 
$
\mathcal{A}_{R_n} := \big\{\max_{1\leq i \leq  n} \|X_i\|\leq {R_n}\big\}.
$ 
On $\mathcal{A}_{R_n}$, any sample points $x, x'$ satisfy $\|x\|,\|x'\|\leq {R_n}$, hence
$$
\|x-x'\|\leq \|x\|+ \|x'\| \leq 2{R_n}.
$$
With \Cref{ass:tail}, we have
\begin{align*}
    q_{R_n} =\Pr(\mathcal{A}^c_{R_n}) = \Pr(\exists i,\|X_i\|>{R_n})&\leq n C\exp(-\kappa R_n)\\
    &=n C\exp\Big(-\kappa \frac{3\log n}{\kappa}\Big)\\
    &= C/n^2.
\end{align*}

We adapt the bounded-difference argument of \citet{briol19inference}. Denote $h(x_1,\ldots,x_n) = \ed_U(\P,\P^n)$, explicitly we have
\begin{align*}
    h(x_1, \ldots, x_n)
    &=\frac{2}{n}\sum_{i=1}^n\E_{X\sim\P}\|X-x_i\|- \E_{X,X'\sim \P}\|X-X'\|-\frac{1}{n(n-1)}\sum_{j=1}^n\sum_{\substack{k=1\\k\neq j}}^n\|x_j-x_k\|.
\end{align*}
As $\mathcal{A}_{R_n}=\bigcap_{i=1}^n\{\|X_i\|\le R_n\}$ is a product event, conditioning on it keeps the $X_i$ independent and supported in $\{\|x\|\le R_n\}$; we apply McDiarmid in this conditional distribution, where the bounded difference need only hold for $x_i,x_i'$ in the ball. Only the first and last term can be affected by changing a single point from $x_i$ to $x_i'$. On the event $\mathcal{A}_{R_n}$, the first term changes by at most
\begin{align*}
    \frac{2}{n}\Big|\E\|X-x_i\|-\E\|X-x'_i\|\Big|
    \leq\frac{2}{n}\|x_i-x'_i\|
    \leq \frac{2}{n}\big(\|x_i\|+\|x'_i\|\big)
    \leq \frac{4R_n}{n},
\end{align*}
using $|\E Z|\le\E|Z|$ and $\|x_i\|,\|x'_i\|\le R_n$. The last term $\frac{1}{n(n-1)}\sum_{j\neq k}\|x_j-x_k\|$ changes only in the $2(n-1)$ pairs involving $x_i$, each by at most $\|x_i-x'_i\|\le 2R_n$, hence by at most $\frac{1}{n(n-1)}\cdot 2(n-1)\cdot 2R_n=\frac{4R_n}{n}$. The bounded-difference constant is therefore $8R_n/n$.

By McDiarmid's inequality \citep{mcdiarmid1989method}, we get that for any $\epsilon>0$,
\begin{align*}
    P\Big(\ed_U(\P,\P^n)-\E\big[\ed_U(\P,\P^n)\,|\,\mathcal{A}_{R_n}\big]\geq\epsilon\,\big|\,\mathcal{A}_{R_n}\Big) &\leq \exp\big(-2\epsilon^2 /\{n(8{R_n}/n)^2\} \big)\\
    &=\exp\Bigg(\frac{-n\epsilon^2}{32{R_n}^2}\Bigg).
\end{align*}
Setting this bound to $\delta$ gives $\epsilon=4{R_n}\sqrt{\frac{2}{n}\log\big(\frac{1
}{\delta}\big)}$. Thus, we get that on the event $\mathcal{A}_{R_n}$
$$
 P\bigg(\ed_U(\P,\P^n)-\E\big[\ed_U(\P,\P^n)\,|\,\mathcal{A}_{R_n}\big]<4{R_n}\sqrt{\frac{2}{n}\log\Big(\frac1\delta\Big)}\bigg) \geq  1-\delta.
$$

Observe that
\begin{align*}
    \E\Big[\ed_U(\P,\P^n)\big|\mathcal{A}_{R_n}\Big] &= \frac{\E[\ed_U(\P,\P^n)\mathbbm{1}(\mathcal{A}_{R_n})]}{\Pr(\mathcal{A}_{R_n})}\\
    &\leq \frac{\sqrt{\E [\ed_U(\P,\P^n)^2]}\sqrt{\Pr(\mathcal{A}_{R_n})}}{\Pr(\mathcal{A}_{R_n})}\\
    &\leq \frac{\sqrt{\E [\ed_U(\P,\P^n)^2]}}{\sqrt{1-q_{R_n}}}\\
    &\leq\frac{C_1M}{\sqrt{n(1-q_{R_n})}};
\end{align*}
here the last inequality is from \Cref{lemma:finite_ed_second}.

We thus obtain that conditioning on event $\mathcal{A}_{R_n}$:
\begin{align*}
    P\bigg(\ed_U(\P,\P^n)<\frac{C_1M}{\sqrt{n(1-q_{R_n})}}+4{R_n}\sqrt{\frac{2}{n}\log\Big(\frac1\delta\Big)}\bigg) &\geq 1-\delta,
\end{align*}
and thus the unconditional probability is
\begin{align*}
    P\bigg(\ed_U(\P,\P^n)<\frac{C_1M}{\sqrt{n(1-q_{R_n})}}+4{R_n}\sqrt{\frac{2}{n}\log\Big(\frac1\delta\Big)}\bigg)\geq(1-\delta)\Pr(\mathcal{A}_{R_n}) = (1-\delta)(1-q_{R_n}).
\end{align*}
The claim follows.
\end{proof}

\Cref{lemma:concentration_inequality} bounds $\ed_U(\P,\P^n)$; it transfers to $\ed(\P,\P^n)$ because the two differ by $U_n/n$, with $U_n=\frac{1}{n(n-1)}\sum_{j\neq k}\|x_j-x_k\|$. On $\mathcal A_{R_n}$ each $\|x_j-x_k\|\le 2R_n$, so $U_n/n\le 2R_n/n$, lower-order than the $R_n\sqrt{n^{-1}\log(1/\delta)}$ term.

\begin{lemma}[Energy-distance perturbation]
\label{lemma:ed_perturbation}
For probability measures $P,Q,R$ on $\R^d$,
\[
|\ed(P,Q)-\ed(P,R)|
\,\leq\,\big(\sqrt{\ed(P,Q)}+\sqrt{\ed(P,R)}\big)\sqrt{\ed(Q,R)}.
\]
\end{lemma}
\begin{proof}
With $\ed(P,Q)=\|\mu_P-\mu_Q\|^2_{\mathcal H}$ in the RKHS of the energy
kernel,
\[
\ed(P,Q)-\ed(P,R)
= \big\langle\mu_R-\mu_Q,\,2\mu_P-\mu_Q-\mu_R\big\rangle.
\]
Cauchy--Schwarz with $\|2\mu_P-\mu_Q-\mu_R\|\leq\|\mu_P-\mu_Q\|+\|\mu_P-\mu_R\|$ concludes the claim.
\end{proof}

\begin{lemma}[Uniform convergence over a bounded parameter set]
\label{lemma:uniform_ed}
Let $\Theta_0\subseteq\Theta$ be bounded with closure $\overline{\Theta_0}$, and suppose \Cref{ass:tail} holds and, for each $\nu$, the map $\theta\mapsto \textsl{g}_\theta(\nu)$ is Lipschitz on $\overline{\Theta_0}$, that is, $\|\textsl{g}_\theta(\nu)-\textsl{g}_{\theta'}(\nu)\|\le L(\nu)\,\|\theta-\theta'\|$, where $L(\nu)$ satisfies $\E[\sqrt{L(\nu)}]<\infty$. With $\P^m_\theta=\frac1m\sum_{i=1}^m\delta_{\textsl{g}_\theta(\nu_i)}$ for $\nu_1,\dots,\nu_m\overset{\rm iid}{\sim}\cN(0,I)$,
\begin{align*}
\sup_{\theta\in\overline{\Theta_0}}\ed(\P_\theta,\P^m_\theta)\to0\quad\text{in probability.}
\end{align*}
\end{lemma}
\begin{proof}
$\overline{\Theta_0}$ is compact, and by \Cref{lemma:concentration_inequality}, $\ed(\P_\theta,\P^m_\theta)\to0$ in probability for each fixed $\theta$, at a rate set by the common constants of \Cref{ass:tail}. It remains to make this convergence uniform over $\theta$.

Write $D_m(\theta)=\sqrt{\ed(\P_\theta,\P^m_\theta)}$. As $\sqrt{\ed}=\|\mu_\cdot-\mu_\cdot\|_{\mathcal H}$ is a metric,
\begin{align*}
\big|D_m(\theta)-D_m(\theta')\big|\le \sqrt{\ed(\P_\theta,\P_{\theta'})}+\sqrt{\ed(\P^m_\theta,\P^m_{\theta'})}.
\end{align*}
For the energy kernel $\|k(\cdot,a)-k(\cdot,b)\|_{\mathcal H}=\sqrt{\|a-b\|}$, so
\begin{align*}
\sqrt{\ed(\P_\theta,\P_{\theta'})}\le\E_\nu\sqrt{\|\textsl{g}_\theta(\nu)-\textsl{g}_{\theta'}(\nu)\|}
\le\sqrt{\|\theta-\theta'\|}\,\E[\sqrt{L(\nu)}],
\end{align*}
the last step by the Lipschitz condition; the same bound, with $\E_\nu$ replaced by $\frac1m\sum_{i=1}^m$, holds for $\sqrt{\ed(\P^m_\theta,\P^m_{\theta'})}$, since the two empirical measures share the draws $\nu_1,\dots,\nu_m$. Hence $|D_m(\theta)-D_m(\theta')|\le A_m\sqrt{\|\theta-\theta'\|}$ with $A_m=\E[\sqrt{L(\nu)}]+\frac1m\sum_{i=1}^m\sqrt{L(\nu_i)}$, and $A_m\to 2\,\E[\sqrt{L(\nu)}]$ a.s., so $A_m=O_P(1)$.

Fix $\eta,\gamma>0$ and take $A<\infty$ with $\Pr(A_m\le A)\ge 1-\gamma$ for all large $m$. On $\{A_m\le A\}$ the modulus bound gives $|D_m(\theta)-D_m(\theta')|\le\eta$ whenever $\|\theta-\theta'\|\le(\eta/A)^2$; covering $\overline{\Theta_0}$ by finitely many such balls, centred at $\theta_1,\dots,\theta_N$,
\begin{align*}
\sup_{\theta\in\overline{\Theta_0}}D_m(\theta)\le\max_{1\le j\le N}D_m(\theta_j)+\eta.
\end{align*}
The finite maximum tends to $0$ in probability, and $\eta,\gamma$ are arbitrary, so $\sup_\theta D_m(\theta)\to0$, and hence $\sup_\theta\ed(\P_\theta,\P^m_\theta)\to0$, in probability.
\end{proof}

\begin{theorem}[Generalisation bounds]
\label{thm:generalization_bound}
Write 
\begin{align*}
\hat\theta_n=\argmin_{\theta\in\Theta}\ed_U(\P_\theta,\Q^n) \quad \mbox{and}\quad \hat\theta_{m,n}=\argmin_{\theta\in\Theta}\ed_U(\P^m_\theta,\Q^n).
\end{align*} Under
the assumptions of \Cref{lemma:concentration_inequality}, there exists
a deterministic sequence $K_n=O(1+R_n)$ such that $\ed(\P_\theta,\Q)$ and $\ed(\P_\theta,\Q^n)$ are bounded by $K_n$, uniformly in
$\theta\in\Theta$, on the truncation event $\mathcal A_{R_n}$. With
probability at least $1-\delta-q_{R_n}$,
\[
\ed(\P_{\hat\theta_n},\Q)-\inf_{\theta\in\Theta}\ed(\P_\theta,\Q)
\,\leq\, 4\sqrt{K_n}\,\sqrt{\frac{C_1M}{\sqrt{n(1-q_{R_n})}}+4R_n\sqrt{\frac{2}{n}\log\tfrac{1}{\delta}}}.
\]

For the estimator $\hat\theta_{m,n}$, under conditions (c)--(d) of \Cref{thm:convergence},
\[
\ed(\P_{\hat\theta_{m,n}},\Q)-\inf_{\theta\in\Theta}\ed(\P_\theta,\Q)\to0\qquad\text{in probability as }m,n\to\infty.
\]
\end{theorem}

\begin{proof}[Proof of \Cref{thm:generalization_bound}]
On $\mathcal A_{R_n}$, decompose
\[
\ed(\P_{\hat\theta_n},\Q)-\inf_{\theta\in\Theta}\ed(\P_\theta,\Q)
\;=\; T_1+T_2+T_3,
\]
with $T_1=\ed(\P_{\hat\theta_n},\Q)-\ed(\P_{\hat\theta_n},\Q^n)$,
$T_2=\ed(\P_{\hat\theta_n},\Q^n)-\inf_{\theta\in\Theta}\ed(\P_\theta,\Q^n)$
and
$T_3=\inf_{\theta\in\Theta}\ed(\P_\theta,\Q^n)-\inf_{\theta\in\Theta}\ed(\P_\theta,\Q)$.
The optimality of $\hat\theta_n$ gives $T_2=0$: $\ed_U(\P_\theta,\Q^n)$ and $\ed(\P_\theta,\Q^n)$ differ only by the $\Q^n$ self-term, constant in $\theta$, so they share the minimiser $\hat\theta_n$.
\Cref{lemma:ed_perturbation} bounds
$|T_1|\leq 2\sqrt{K_n}\,\sqrt{\ed(\Q,\Q^n)}$; the same lemma combined with
$\inf_\theta f(\theta)-\inf_\theta g(\theta)\leq\sup_\theta(f(\theta)-g(\theta))$
gives $|T_3|\leq 2\sqrt{K_n}\,\sqrt{\ed(\Q,\Q^n)}$. Hence
\[
\ed(\P_{\hat\theta_n},\Q)-\inf_{\theta\in\Theta}\ed(\P_\theta,\Q)
\,\leq\, 4\sqrt{K_n}\,\sqrt{\ed(\Q,\Q^n)}\qquad\text{on }\mathcal A_{R_n},
\]
and \Cref{lemma:concentration_inequality} delivers the first display in the
statement.

For $\hat\theta_{m,n}$, condition (c) makes $\Theta$ compact, so $\hat\theta_{m,n}\in\Theta$. By optimality of $\hat\theta_{m,n}$,
\begin{align*}
\ed(\P_{\hat\theta_{m,n}},\Q)-\inf_{\theta\in\Theta}\ed(\P_\theta,\Q)
\le 2\sup_{\theta\in\Theta}\big|\ed_U(\P^m_\theta,\Q^n)-\ed(\P_\theta,\Q)\big|.
\end{align*}
Here $\ed_U(\P^m_\theta,\Q^n)$ and $\ed(\P^m_\theta,\Q^n)$ differ by the self-term correction, $o_P(1)$ uniformly in $\theta$ under conditions (c)--(d), so we may pass to $\ed$ and apply the triangle inequality. Write $d=\sqrt{\ed}$; then $|d(\P^m_\theta,\Q^n)-d(\P_\theta,\Q)|\le d(\P^m_\theta,\P_\theta)+d(\Q,\Q^n)$, so
\begin{align*}
\big|\ed(\P^m_\theta,\Q^n)-\ed(\P_\theta,\Q)\big|\le\big(d(\P^m_\theta,\P_\theta)+d(\Q,\Q^n)\big)\big(d(\P^m_\theta,\Q^n)+d(\P_\theta,\Q)\big).
\end{align*}
The first factor is $o_P(1)$ uniformly in $\theta$, by \Cref{lemma:uniform_ed} (with $\Theta_0=\Theta$) and \Cref{lemma:concentration_inequality}. For the second, $\sup_{\theta\in\Theta}d(\P_\theta,\Q)=O(1)$ by \Cref{lem:finite_moment}, whence $d(\P^m_\theta,\Q^n)\le d(\P^m_\theta,\P_\theta)+d(\P_\theta,\Q)+d(\Q,\Q^n)$ is $O_P(1)$ uniformly in $\theta$. Hence $\sup_{\theta\in\Theta}|\ed(\P^m_\theta,\Q^n)-\ed(\P_\theta,\Q)|=o_P(1)$, and the excess risk tends to $0$ in probability.
\end{proof}


With the above theorems, we can now prove \Cref{thm:convergence}.

\begin{proof}[Proof of \Cref{thm:convergence}]
By \Cref{thm:generalization_bound} applied with $\delta=1/n^2$,
\[
\Pr\big\{\ed(\P_{\hat\theta_n},\Q)-\inf_{\theta\in\Theta}\ed(\P_\theta,\Q)
\geq 4\sqrt{K_n}\,\sqrt{B_n}\big\}
\leq \tfrac{1}{n^2}+q_{R_n}=O\big(\tfrac{1}{n^2}\big),
\]
where
$B_n=\frac{C_1M}{\sqrt{n(1-q_{R_n})}}+8R_n\sqrt{\frac{\log n}{n}}=O\big(\frac{(\log n)^{3/2}}{\sqrt n}\big)$.
Since $K_n=O(\log n)$, we have $\sqrt{K_nB_n}\to0$, and the Borel--Cantelli lemma then gives, almost surely,
\[
\ed(\P_{\hat\theta_n},\Q)-\inf_{\theta\in\Theta}\ed(\P_\theta,\Q)\,\to\,0.
\]
When $\P_{\theta^\ast}\stackrel{d}{=}\Q$, the infimum equals $0$, so
$\ed(\P_{\hat\theta_n},\Q)\xrightarrow{P}0$, and
\Cref{thm:convergence_by_convergent_ed} gives
$\P_{\hat\theta_n}\Rightarrow\Q$ in probability. For $\hat\theta_{m,n}$, under conditions (c)--(d) of \Cref{thm:convergence}, the
$(m,n)$ part of \Cref{thm:generalization_bound} gives
$\ed(\P_{\hat\theta_{m,n}},\Q)\xrightarrow{P}0$ under the same specification, and
\Cref{thm:convergence_by_convergent_ed} again yields
$\P_{\hat\theta_{m,n}}\Rightarrow\Q$ in probability.
\end{proof}

Based on these results, we provide proof of \Cref{cor:component_convergence}.
\begin{proof}
For the first claim, apply \Cref{thm:convergence} to \eqref{eq:empirical_g}
with target $\P_{ZX}$; \Cref{thm:convergence_by_convergent_ed} yields
$\hat g(\epsilon)\Rightarrow\P_{ZX}$ in probability.

The second and third claims both follow once the conclusions of \Cref{prop:population_estimates} are shown to hold, off \emph{common} null sets, for \emph{every} population minimiser rather than for a single solution. The proof of \Cref{prop:population_estimates} uses strict propriety, hence characterises any minimiser: each $\theta^\ast\in\Theta^\ast$ satisfies $f_{\theta^\ast}\big(x,h_{\theta^\ast}(z,x,\xi)\big)\sim\P_{Y\,|\, Z=z,X=x}$ off some $\P_{ZX}$-null set, and $h_{\theta^\ast}(Z_x,x,\xi)\sim\cN(0,I_{d_y})$ with $Z_x\sim\P_{Z\mid X_0=x_0}$ off some $\P_X$-null set, both sets depending on $\theta^\ast$. As $\Theta$ is compact, $\Theta^\ast$ is separable; fix a countable dense $\{\theta_k\}\subset\Theta^\ast$ and let $N$ and $M$ be the unions of the corresponding $\P_{ZX}$- and $\P_X$-null sets, themselves null. For any $\theta^\ast\in\Theta^\ast$ take $\theta_k\to\theta^\ast$; since $\theta\mapsto\P_\theta$ is continuous and the distributions indexed by the $\theta_k$ agree, off $N$ and $M$, with the same true conditional and standard Gaussian, uniqueness of weak limits gives the same two properties for $\theta^\ast$ off the common null sets $N$ and $M$.

We assume the empirical objective \eqref{eq:empirical_fh} converges uniformly over $\Theta$ to its population counterpart; when $X_0$ is present, this also requires the auxiliary model used to draw $Z$ from $\P_{Z\mid X_0}$ to be consistent. Under the compactness, componentwise Lipschitz, and uniform tail conditions of \Cref{thm:convergence} together with this consistency, the excess population risk of $\hat\theta$ tends to zero. Since $\Theta$ is compact and the population risk is continuous, every limit point of $\hat\theta$ lies in $\Theta^\ast$, so $\operatorname{dist}(\hat\theta,\Theta^\ast)\to 0$. For the second claim, fix $(z,x)\notin N$: the conditional fit then holds for every $\theta^\ast\in\Theta^\ast$, so by continuity of $\theta\mapsto\P_\theta$, along any subsequence with $\hat\theta\to\theta^\ast$ we have $\hat f(x,\hat h(z,x,\xi))\Rightarrow\P_{Y\,|\, Z=z,X=x}$; as this limit is the same for all $\theta^\ast\in\Theta^\ast$, it holds in probability for $\P_{ZX}$-almost every $(z,x)$.

For the third claim, fix $x\notin M$ and outside the positivity-null set. By positivity (\Cref{ass:identification}) the conditional fit holds at $\P_{Z\mid X_0=x_0}$-almost every $z$, and the calibration gives $h_{\theta^\ast}(Z_x,x,\xi)\sim\cN(0,I_{d_y})$; substituting both into the identification formula of \Cref{lem:identifiability} yields $f_{\theta^\ast}(x,\eta)\sim\P_{\Yx}$ for every $\theta^\ast\in\Theta^\ast$. Continuity of $\theta\mapsto\P_\theta$ then gives, along any subsequence with $\hat\theta\to\theta^\ast$, $\hat f(x,\eta)\Rightarrow f_{\theta^\ast}(x,\eta)\sim\P_{\Yx}$; since this limit is the same for all $\theta^\ast\in\Theta^\ast$, $\hat f(x,\eta)\Rightarrow\P_{\Yx}$ in probability for $\P_X$-almost every $x$.

\end{proof}

\subsection{Proof of \Cref{prop:extrap}}
\label{sec:proof_extrap}
\begin{proof}
    Since $\widetilde\eta$ is a continuous random variable, there exists a strictly monotone function $\widetilde{f}_1$ such that $\widetilde\eta\overset{d}{=} \widetilde{f}_1(\widetilde\xi)$ where $\widetilde\xi\sim\cN(0,1)$. Let $f^\ast,f^\ast_1$ denotes the element of an optimal frengression model such that
    \[
    f^\ast(x+f^\ast_1(\widetilde\xi))\sim \P_{\Yx},
    \]
    for $\P_X$-almost every $x\in\cX$ by \Cref{prop:population_estimates}. Under \Cref{ass:preanm} both sides depend continuously on $x$ and $\P_X$ has full support on $\cX$, so the relation extends to every $x\in\cX$. This implies
    \[
    \widetilde{f}(x+\widetilde{f}_1(\widetilde\xi))\overset{d}{=}f^\ast(x+f^\ast_1(\widetilde\xi)),
    \]
    for all $x\in\cX$. Taking medians in the last display gives $\widetilde{f}(x)=f^\ast(x)$ on $\cX$, so $\widetilde{f}$ and $f^\ast$ are increasing or decreasing together; replacing $Y$ by $-Y$ in the decreasing case, and choosing $\widetilde{f}_1$ and $f^\ast_1$ increasing, which is possible because $\widetilde\xi\overset{d}{=}-\widetilde\xi$, we may take all four maps strictly increasing.

    Since both sides are strictly increasing functions of the same $\widetilde\xi\sim\cN(0,1)$, equality in distribution implies that they coincide as functions: for all $x\in\cX$ and $u\in\R$,
    \begin{equation}\label{eq:pf2}
        \widetilde{f}(x+\widetilde{f}_1(u))=f^\ast(x+f^\ast_1(u)).
    \end{equation}
    The maps $\widetilde f,\widetilde f_1, f^\ast, f^\ast_1$ are continuously differentiable by \Cref{ass:preanm}; taking the derivative w.r.t.~$x$ and $u$ respectively yields
    \begin{align*}
    \widetilde{f}'(x+\widetilde{f}_1(u)) &= {f^\ast}'(x+f^\ast_1(u))\\
    \widetilde{f}'(x+\widetilde{f}_1(u))\widetilde{f}'_1(u) &= {f^\ast}'(x+f^\ast_1(u)){f^\ast_1}'(u).
    \end{align*}
    Using the first equation to replace ${f^\ast}'(x+f^\ast_1(u))$ in the second gives
    \[
    \widetilde{f}'(x+\widetilde{f}_1(u))\big(\widetilde{f}'_1(u)-{f^\ast_1}'(u)\big)=0.
    \]
    Since $\widetilde{f}$ is strictly increasing, its derivative cannot be zero everywhere on the interval $\{x+\widetilde{f}_1(u):x\in[x_{\rm min},x_{\rm max}]\}$ -- otherwise $\widetilde{f}$ would be constant there. Hence for each $u$ there is an $x$ with $\widetilde{f}'(x+\widetilde{f}_1(u))>0$, so $\widetilde{f}'_1(u)={f^\ast_1}'(u)$. As this holds for every $u$, integration gives $\widetilde{f}_1(u)=f^\ast_1(u)+c$ for a constant $c$. Since $f^\ast_1(0)=0$ by normalisation and $\widetilde{f}_1(0)=0$ because $\widetilde\eta$ has median $0$, we obtain $c=0$ and $\widetilde{f}_1=f^\ast_1$.

    Then according to \eqref{eq:pf2}, we have 
    \[
    \widetilde{f}(\widetilde{x}) = f^\ast(\widetilde{x}),
    \]
    for all $\widetilde{x}\in[x_{\rm min}-\widetilde\eta_0, x_{\rm max}+\widetilde\eta_0]$. Note that $\widetilde{f}(x)$ and $f^\ast(x)$ are exactly the true and modelled conditional medians, respectively; the desired result then follows.
\end{proof}



%
%

\section{Hyperparameters in frengression}
All hyperparameters are kept fixed across experiments, except for the number of epochs, which is chosen so that the loss converges. The hyperparameters include
\begin{itemize}

    \item learning\_rate: The learning rate is set to be $10^{-4}$.
    \item hidden\_layer: The number of hidden layers in the neural network is set to be $3$.
    \item hidden\_dim: The size of neurons per hidden layer is set to be $100$.
\end{itemize}

\section{Experiment details}
All experiments were conducted on a MacBook with an Apple M3 chip, 8-core CPU, and 32GB RAM.
\subsection{Binary intervention}
\label{sec:binary_exp_details}
In this setting, for each simulation, we train frengression for $1000$ epochs.

For the AIPW estimator, we fit the propensity score model with logistic regression; because the true propensity scores are linear in the covariates, this specification performs well. We model the outcomes separately using random forests.

Next, we compare frengression to two advanced deep neural network approaches for causal inference, CausalEGM and Dragonnet. A comprehensive list of CausalEGM’s hyper-parameters is available at \url{https://causalegm.readthedocs.io/en/latest/tutorial_py.html}. Because the search space is extensive, we adopt the configuration recommended for the binary treatment experiments in the original paper. The settings of key hyperparameters are given in \Cref{tab:egm_params}; all other quantities remain at their default values. Although the authors of CausalEGM suggest  30,000 epochs especially for continuous treatment setting, for fair comparison we train CausalEGM for $1000$ epochs in our experiments.
\begin{table}[ht]
  \centering
  \begin{tabular}{@{}ll@{}}
    \toprule
    Hyperparameter & Value \\ \midrule
    Latent covariate dimensions: z\_dims               & ${[}1, 1, 1, 1{]}$ \\
    Learning rate: lr                        & $2\times10^{-4}$ \\
    Reconstruction loss weight: alpha               & 1 \\
    Round trip loss weight: beta                            & 1 \\
    Gradient penalty loss weight: gamma                  & 10 \\
    Generator units: g\_units                              & {[}64, 64, 64, 64, 64{]} \\
    Encoder units: e\_units                                & {[}64, 64, 64, 64, 64{]} \\
    F network units: f\_units                            & {[}64, 32, 8{]}  \\
    H network units: h\_units                            & {[}64, 32, 8{]}  \\
    Discriminator units in latent space: dz\_units                   & {[}64, 32, 8{]}  \\
    Discriminator units in covariate space: dv\_units                   & {[}64, 32, 8{]}  \\ \bottomrule
  \end{tabular}
    \caption{Values of key hyperparameters in the CausalEGM model.}
      \label{tab:egm_params}
\end{table}

For Dragonnet, in line with its requirements, we tuned the hyperparameters with \texttt{optuna}. The search range is shown in \Cref{tab:dragonnet_params}. We conduct $30$ hyperparameter search trials, selecting the best configuration by validation performance on a hold-out set of $400$ samples, and then apply the resulting model to an independent  test set of $400$ samples to obtain the reported estimates.

\begin{table}[ht]
  \centering
  \begin{tabular}{@{}ll@{}}
    \toprule
    Hyperparameter & Value \\ \midrule

    Learning rate: lr                        & $\operatorname{logUniform}(10^{-5},10^{-2})$ \\
    Weights decay rate: wd              & $\operatorname{logUniform}(10^{-5},10^{-2})$ \\
    Batch size: bs                          & $\{32,128,256\}$ \\
    Number of iterations: epochs                  & Grid search in between $200$ and $600$ \\
    Number of neurons in shared layers: shared\_hidden                              & Grid search in between $50$ and $200$\\
    Number of neurons in shared layers: outcome\_hidden                              & Grid search in between $50$ and $200$  \\ \bottomrule
  \end{tabular}
    \caption{Search ranges of key hyperparameters in the Dragonnet model.}
      \label{tab:dragonnet_params}
\end{table}

\subsection{Longitudinal Analysis}
\subsubsection{Synthetic Data}
\label{sec:sequential_exp_details}
We provide the details of the synthetic data generation in \Cref{sec:longitudinal}.

\begin{table}[ht]
\centering
\small                     
\setlength{\tabcolsep}{2pt}
\resizebox{\textwidth}{!}{%
\begin{tabular}{llll}
\toprule
 & \textbf{Setting 1} & \textbf{Setting 2} & \textbf{Setting 3}\\
\midrule
$C$ & $\operatorname{Bernoulli}(0.5)$ & $\operatorname{Exp}(1)$ & $\operatorname{Bernoulli}(0.5)$\\[2pt]

$Z_{t}\mid\overline{X}_{t-1},C$ &
  $\mathcal N(-0.5+0.5X_{t-1}+0.25C,\,0.5)$ &
  $\mathcal N(0.7Z_{t-1}+0.2X_{t-1})$ &
  $\mathcal N(-0.5+0.5X_{t-1}+0.25C,\,0.5)$\\[2pt]

$X_{t}\mid\overline{Z}_{t},C$ &
  $\operatorname{Bernoulli}\big(\operatorname{expit}(0.2Z_{t}+0.1)\big)$ &
  $\operatorname{Bernoulli}\big(\operatorname{expit}(-0.5+0.25X_{t-1}+0.5Z_{t})\big)$ &
  $\mathcal N(1.2+0.1Z_{t}+2C,\,0.1)$\\[2pt]

$\widetilde Y_{t}(\overline{x}_{t})\mid C$ &
  $\operatorname{Exp}(0.3+0.2x_{t}+0.1C)$ &
  $\operatorname{Exp}(0.5+0.2x_{t}+0.2C,1)$ &
  $\operatorname{Exp}(0.1+0.3x_{t}+0.1C)$\\[2pt]

\makecell[l]{Gaussian copula \\  (with Pearson \\ correlation $\rho$)}  &
  \makecell[l]{$\rho=0.4$ between\\ $Z_{t}$ and $Y_{t}$} &
  \makecell[l]{pair-copula:\\ $\rho_{Y_{t},Z_{t-1}}=0.2$,\\ $\rho_{Y_{t},Z_{t}\mid Z_{t-1}}=0.3$\\ (depends on $C$)} &
  \makecell[l]{$\rho=0.4$ between\\ $Z_{t}$ and $Y_{t}$}\\
\bottomrule
\end{tabular}
}
\caption{Simulation settings 1--3.}
\end{table}

\begin{table}[ht]
\centering
\small
\setlength{\tabcolsep}{6pt} 
\begin{tabular}{lc}
\toprule
 & \textbf{Setting 4}\\
\midrule
$C$ & $\mathcal N(0,1)$\\[2pt]

$Z_{t}\mid\overline{X}_{t-1},C$ &
  $\mathcal N(-0.5+0.5X_{t-1}+0.5Z_{t-1}+0.5C,\,0.5)$\\[2pt]

$X_{t}\mid\overline{Z}_{t},C$ &
  $\mathcal N(2Z_{t}+0.1C,\,1)$\\[2pt]

$Y_{t}(\overline{x}_{t})\mid C$ &
  $\mathcal N(2x_{t}+x_{t-1}+0.5x_{t-2}+C,\,1)$\\[2pt]
  \makecell[l]{Gaussian copula \\  with Pearson correlation $\rho$}  & $\rho=0.24$\\[2pt]
\bottomrule
\end{tabular}
\caption{Simulation setting 4.}
\end{table}

\subsubsection{Survival Analysis: comparison with \texttt{dltmle}}
\label{sec:survival_dltmle}
In our implementation of \texttt{dltmle} \citep{dltmle} in \Cref{sec:longitudinal}, we tuned the hyperparameters using their built in \texttt{dltmle.tune} function within the candidates in \Cref{tab:dltmle_params}.
\begin{table}[ht]
  \centering
  \begin{tabular}{@{}ll@{}}
    \toprule
    Hyperparameter & Hyperparameter values / candidates for search \\ \midrule
    dim\_emb               & $\{8,16\}$ \\
    dim\_emb\_time                   & $\{4,8\}$ \\
    dim\_emb\_type               & $\{4,8\}$ \\
    hidden\_size              & $\{8,16,32\}$ \\
    num\_layers                  & $\{1,2,4\}$ \\
    nhead          & $\{2,4\}$ \\
    dropout    & $\{0,0.1,0.2\}$ \\
    learning\_rate  & $\{10^{-3},5\times 10^{-4}, 10^{-4}, 5\times 10^{-3}\}$  \\
    alpha      & {[}0.05, 0.1, 0.5, 1{]}  \\
    beta      & {[}0.05, 0.1, 0.5, 1{]}  \\
 \bottomrule
  \end{tabular}
    \caption{Search ranges of key hyperparameters of  \texttt{dltmle} model.}
      \label{tab:dltmle_params}
\end{table}

\subsection{LEADER}
\label{sec:LEADER_experiment_details}
We provide the description of each variable name in our experiment on the LEADER trial in \Cref{tab:leader_description}.

\begin{table}[ht]
  \centering
  \begin{tabular}{ll}
    \toprule
    Variable & Description \\ 
    \midrule
    SEX\_F & Binary variable. If \texttt{True}, the patient is female.\\
    SMOKER\_NEVER.SMOKED & Binary variable. If \texttt{True}, the patient has never smoked before.\\
    STENFL   & Carotid $>50\%$ stenosis on angiography. Binary. \\
    NEPSCRFL &    Diabetic Nephropathy at screening flag. Binary.\\
    Age & Patient's age when enrolled in the trial. \\
    DIABDUR & Diabetes duration (years) at baseline.\\
    HDL1BL & 	HDL Cholesterol (mmol/L) at baseline.\\
    LDL1BL &  LDL Cholesterol (mmol/L) at baseline. \\
    CHOL1BL & Total Cholesterol (mmol/L) at baseline.\\
    log\_TRIG1BL & Logarithm of triglycerides (mmol/L) at baseline. \\
    log\_CREATBL & 	Serum Creatinine ($\mu$mol/L) at baseline.\\
 \bottomrule
  \end{tabular}
    \caption{Description of each variable used in \Cref{sec:leader}.}
      \label{tab:leader_description}
\end{table}

The hyperparameter configurations used in \texttt{dltmle} are listed in \Cref{tab:dltmle_leader_params}. Similar to its implementation in \Cref{sec:survival_dltmle}, we tuned the hyperparameters with the built in \texttt{dltmle.tune} function over $20$ trials.

\begin{table}[ht]
  \centering
  \begin{tabular}{@{}ll@{}}
    \toprule
    Hyperparameter & Hyperparameter values / candidates for search \\ \midrule
    dim\_emb               & $\{16,32\}$ \\
    hidden\_size              & $\{32,64\}$ \\
    num\_layers                  & $\{1,2,4\}$ \\
    nhead          & $\{2,3,4\}$ \\
    dropout    & $\{0,0.1,0.2\}$ \\
    learning\_rate  & $\{10^{-3},10^{-2}\}$  \\
    batch\_size      & $\{64,128\}$  \\
    epochs      & $[50,200]$  \\
 \bottomrule
  \end{tabular}
    \caption{Search ranges of key hyperparameters of  \texttt{dltmle} model on the LEADER trial.}
      \label{tab:dltmle_leader_params}
\end{table}
\bibliography{paper-ref}
\end{document}